\documentclass[sn-mathphys,Numbered]{sn-jnl}

\usepackage{graphicx}
\usepackage{multirow}%
\usepackage{amsmath,amssymb,amsfonts}%
\usepackage{amsthm}%
\usepackage{mathrsfs}%
\usepackage[title]{appendix}%
\usepackage{xcolor}%
\usepackage{textcomp}%
\usepackage{manyfoot}%
\usepackage{booktabs}%
\usepackage{algorithm}%
\usepackage{algorithmicx}%
\usepackage{algpseudocode}%
\usepackage{listings}%
\usepackage{float}%
\usepackage{siunitx}
\usepackage{tabularx}
\usepackage{ltablex}
\usepackage{caption}
\usepackage{makecell}
\usepackage{diagbox}
\usepackage{bm}
\newcommand\setrow[1]{\gdef\rowmac{#1}#1\ignorespaces}
\newcommand\clearrow{\global\let\rowmac\relax}
\clearrow

\newtheorem{theorem}{Theorem}
\newtheorem{lemma}{Lemma}

\newtheorem{example}{Example}
\newtheorem{definition}{Definition}

\begin{document}
	
	\title{Two Sequence-Form Interior-Point Differentiable Path-Following Method to Compute Nash Equilibria}
	
	\author[1]{\fnm{Yuqing} \sur{Hou}}\email{yuqinghou2-c@my.cityu.edu.hk}
	
	\affil[1]{\orgdiv{Department of Automation}, \orgname{University of Science and Technology of China}, \orgaddress{\city{Hefei}, \country{China}}}
	
	
	\abstract{Nash equilibrium is a fundamental solution concept in extensive-form games, while its efficient computation is still far from straightforward. This paper considers finite $n$-player extensive-form games with perfect recall under the sequence-form representation. Unlike existing approaches, which mainly treat the sequence form as a compact computational reformulation, we develop a direct sequence-form definition of Nash equilibrium. Building on this, we rigorously establish the associated sequence-form Nash equilibrium system through an equivalence proof with mixed-strategy Nash equilibrium. On this basis, we propose a single-stage interior-point differentiable path-following method for equilibrium computation. The method uses logarithmic-barrier regularization to generate a differentiable equilibrium path in the interior of the realization-plan space, leading to favorable numerical stability and convergence properties. Numerical results show that the proposed method is effective and computationally efficient.}
	
	
	
	
	\maketitle
	
	\section{Introduction}\label{sec1}
	Extensive-form games~\cite{KuhnExtensiveGames1950} provide a fundamental modeling framework for sequential decision-making problems under uncertainty, and have been widely applied in economics, operations research, and multi-agent systems. By explicitly capturing the temporal structure of decisions, information asymmetry, and strategic interactions among agents, they enable a rich representation of dynamic environments. Within this framework, equilibrium concepts play a central role, as they formalize stable and rational outcomes of strategic behavior. Consequently, the computation of equilibrium solutions is of primary importance for both theoretical analysis and practical applications. Among various equilibrium notions, Nash equilibrium~\cite{NashEquilibriumpointsnperson1950} serves as the cornerstone of game-theoretic analysis, providing a rigorous characterization of mutual best responses among players. This paper focuses on the computation of Nash equilibria for finite $n$-player extensive-form games with perfect recall using the sequence-form representation.
	
	Behavioral-strategy Nash equilibrium and mixed-strategy Nash equilibrium are two formulations of Nash equilibrium in extensive-form games under different strategy representations, and their realization equivalence has already been established. In our previous work~\cite{houSequenceformDifferentiablePathfollowing2025}, we computed the equilibrium system under the sequence-form representation on the basis of behavioral-strategy Nash equilibrium. This approach is also common in the existing literature, where the sequence form is used primarily as an equivalent change of variables for the efficient computation of a target equilibrium, rather than as a representation in which equilibrium is defined directly. However, formulating equilibrium concepts directly in the sequence-form setting is of fundamental importance for understanding their underlying structure, enabling a more rigorous derivation of equilibrium systems, and clarifying the relationships among different equilibrium notions. Taking the normal-form representation as the analytical foundation, this paper extends Nash equilibrium to the sequence-form setting by introducing the notion of sequence-form Nash equilibrium and establishes its equivalence to mixed-strategy Nash equilibrium. Building on this equivalence, we derive the associated equilibrium system in a more rigorous way.
	
	This paper develops a differentiable interior-point path-following framework for computing Nash equilibria in finite extensive-form games under the sequence-form representation. The proposed method regularizes the payoff functions by logarithmic-barrier terms to construct a differentiable equilibrium path in the interior of the realization-plan space, thereby enabling stable numerical tracing toward a Nash equilibrium. Path-following methods have long been recognized as a powerful class of algorithms for computing Nash equilibria in normal-form games. Their origins can be traced to the complementary pivoting procedure of Lemke and Howson~\cite{LemkeEquilibriumPointsBimatrix1964}, which was later extended to $n$-player games by Rosenmüller~\cite{RosenmullerGeneralizationLemkeHowson1971} and Wilson~\cite{WilsonComputingEquilibriaNPerson1971}. Subsequent developments focused on simplicial methods and their refinements, which improved practical implementability~\cite{GarciaSimplicialApproximationEquilibrium1973,vanderLaanComputationFixedPoints1982,Doupnewsimplicialvariable1987,HeringsComputationNashEquilibrium2002}. However, these methods are inherently limited by the lack of smoothness in their underlying formulations, which can significantly hinder convergence efficiency. To overcome this limitation, subsequent research developed differentiable path-following methods, notably the differentiable homotopy method of Herings and Peeters~\cite{Heringsdifferentiablehomotopycompute2001}, the piecewise differentiable global Newton method of Govindan and Wilson~\cite{GovindanglobalNewtonmethod2003}, and the smooth reformulation of Chen and Dang~\cite{Chenreformulationbasedsmoothpathfollowing2016}. Differentiable path-following methods have also been shown to perform well in computing other equilibrium concepts~\cite{Chenextensionquantalresponse2020,Chendifferentiablehomotopymethod2021,caoDifferentiablePathfollowingMethod2023}. In particular, interior-point differentiable path-following methods~\cite{dangInteriorpointPathfollowingAlgorithm2011,dangInteriorPointDifferentiablePathFollowing2022} exhibit significant advantages in terms of numerical stability, convergence behavior, and computational efficiency.
	
	For equilibrium computation in extensive-form games, sequence-form–based path-following methods provide a natural and efficient framework. Introduced by von Stengel~\cite{vonStengelEfficientComputationBehavior1996}, the sequence form yields a compact representation with linear complexity in both variables and constraints, thereby overcoming the dimensional explosion of the normal-form formulation. Building on this representation, a variety of path-following algorithms have been developed, particularly for two-player games~\cite{KollerEfficientComputationEquilibria1996,KollerRepresentationssolutionsgametheoretic1997,VonStengelComputingNormalForm2002,MiltersenComputingquasiperfectequilibrium2010}. For the $n$-player case, Govindan and Wilson~\cite{govindanStructureTheoremsGame2002} developed a piecewise differentiable algorithm for perturbed extensive-form games. More recently, two-stage differentiable path-following methods have been proposed within the sequence-form framework~\cite{houSequenceformDifferentiablePathfollowing2025,houSequenceFormCharacterizationDifferentiable2025}. However, these approaches largely treat the sequence form as a computational tool rather than as a foundation for directly defining and analyzing equilibrium concepts. A preliminary exploration of equivalent definitions of equilibrium in the sequence-form representation has been carried out in~\cite{houCharacterizationComputationNormalForm2026a,zhangComplexityProperEquilibrium2026}, leaving room for a more systematic development. Motivated by this observation, the present paper develops a direct sequence-form treatment of Nash equilibrium for finite $n$-player extensive-form games with perfect recall.
	
	Building on the definition of sequence-form Nash equilibrium introduced in~\cite{houCharacterizationComputationNormalForm2026a}, this paper establishes a rigorously justified sequence-form Nash equilibrium system by proving its equivalence with mixed-strategy Nash equilibrium, and develops an efficient single-stage interior-point differentiable path-following method for its computation. The remaining sections of this paper are organized as follows. Section~\ref{section2} provides a review of Nash equilibrium in extensive-form games and the sequence-form representation.  Section~\ref{section3} rigorously establishes the sequence-form Nash equilibrium system through an equivalence proof with mixed-strategy Nash equilibrium. Section~\ref{section4} presents the proposed single-stage interior-point differentiable path-following methods for computing Nash equilibria. Numerical results and comparative analysis are reported in Section~\ref{section5}, and Section~\ref{section6} concludes the paper.
	
	\section{\large Preliminaries}
	\label{section2}
	\begin{table}[tb!]
		\centering
		\caption{Notation for Extensive-Form Games}
		\begin{tabular}{ll}
			\toprule
			Symbol & Explanation\\
			\midrule
			$N=\{1,2,\ldots,n\}$ & Set of players\\
			$N_c=N\cup\{c\}$ & Set of players and chance player $c$\\
			$a$ & Action taken by a player\\
			$H$ & Set of histories,  $\emptyset\in H$ and $\langle a_1,\ldots,a_L\rangle\in H$ if $\langle a_1,\ldots,a_K\rangle\in H$ and $L<K$\\
			$Z$ & Set of terminal histories\\
			$A(h)=\{a\mid (h,a)\in H\}$ & Set of actions after a nonterminal history $h$\\
			$P(h)$ & Player who takes an action after $h$\\
			$f_{c}(a|h)$ & Probability that chance player $c$ takes action $a$ after $h$\\
			$-i$ & All non-chance players excluding player $i\in N$\\
			$\mathcal{I}_{i}$ & Collection of information partitions of $\{h\in H\mid P(h)=i\}$\\
			$M_{i}=\{1,\ldots,m_{i}\}$ & Set of information partition indices for player $i\in N_c$\\
			$I^{j}_{i}\in\mathcal{I}_{i},j\in M_{i}$ & $j$th information set of player $i\in N_c$, $A(I^j_i)\triangleq A(h)= A(h')$ whenever $h,h'\in I^j_i$\\
			$\succsim_i$ & Preference relation of player $i\in N$ \\
			$u_z^{i}:Z\to\mathbb{R}$ & Payoff function of player $i\in N$\\
			$R_{i}(h)$ & Record of player $i\in N_c$'s experience along $h$\\
			$|C|$ & Cardinality of a finite set $C$\\
			$m_0=\sum_{i\in N}m_i$ & Number of information sets\\
			$n_0=\sum_{i\in N}\sum_{j\in M_i}|A(I^j_i)|$ & Number of actions for non-chance players\\
			\bottomrule
		\end{tabular}
		\label{nfpr-tab-pre1}
	\end{table}
	Following Osborne and Rubinstein~\cite{OsborneCourseGameTheory1994}, an extensive-form game is represented by
	\[
	\Gamma=\langle N, H, P, f_c, \{{\cal I}_i\}_{i\in N}, \{\succsim_i\}_{i\in N}\rangle,
	\]
	where the notation is listed in Table~\ref{nfpr-tab-pre1}. Throughout this paper, we consider finite extensive-form games with perfect recall. Finiteness means that the history set $H$ is finite. Perfect recall requires that, for every player $i\in N_c$, if two histories $h$ and $h'$ belong to the same information set of player $i$, then $R_i(h)=R_i(h')$. The corresponding normal-form representation of $\Gamma$ is =$\Gamma_n=\langle N, S, \sigma^c, \{u^i\}_{i\in N}\rangle$, and the associated notation is provided in Table~\ref{nfpr-tab-pre2}. For each action $a$, define
	\[
	s^i(a)=
	\begin{cases}
		1, & \text{if } s^i(I_i^j)=a \text{ for some $j\in M_i$},\\
		0, & \text{otherwise.}
	\end{cases}
	\]
	Then, for any pure-strategy profile \(s=(s^i:i\in N_c)\), player \(i\)'s payoff is given by
	\begin{equation*}
		\begin{array}{l}
			u^i(s)=\sum\limits_{h=\langle a_1,\ldots,a_L\rangle\in Z}u^i_z(h)\prod\limits_{q=0}^{L-1}s^{P(\langle a_1,\ldots,a_q\rangle)}(a_{q+1}).
		\end{array}
	\end{equation*}
	For a mixed-strategy profile $\sigma=(\sigma^i:i\in N)\in \Xi$, the expected payoff of player $i\in N$ is $u^i(\sigma)=\sum_{s^i\in S^i}\sigma^i(s^i)\,u^i(s^i,\sigma^{-i})$, where
	\begin{equation*}
		\begin{array}{l}
			u^i(s^i,\sigma^{-i})=\sum\limits_{s^{-i}\in S^{-i}}u^i(s^i,s^{-i})\prod\limits_{i_q\in N_c\backslash \{i\}}\sigma^{i_q}(s^{i_q}).
		\end{array}
	\end{equation*}
	
	\begin{definition}\label{nfpe-def-pre1}
		{\em A mixed-strategy profile $\sigma^*$ is called a Nash equilibrium if, for every player $i\in N$, one has $\sigma^{*i}(s^i)=0$ whenever there exists some $\tilde s^i\in S^i$ such that
			\[
			u^i(s^i,\sigma^{*-i})<u^i(\tilde s^i,\sigma^{*-i}).
			\]}
	\end{definition}
	\begin{table}[tb!]
		\centering
		\caption{Notation for Games in Normal Form or Sequence Form}
		\begin{tabular}{ll}
			\toprule
			Symbol & Explanation\\
			\midrule
			$s^i$ & Pure strategy of player $i$\\
			$S=\underset{i\in N_c}{\prod}S^i$ & Set of pure strategy profiles\\
			$\sigma^i$ & Mixed strategy of player $i\in N_c$, probability measure over $S^i$\\
			$\Xi=\underset{i\in N}{\prod}\Xi^i$ & Set of mixed strategy profiles,  $\Xi^i=\{\sigma^i:S^i\to\mathbb{R}_+\mid \sum\limits_{s^i\in S^i}\sigma^{i}(s^i)=1\}$\\
			$\Xi_{++}=\underset{i\in N}{\prod} \Xi^i_{++}$ & Set of strictly positive mixed strategy profiles\\
			$u^i(s)$ & Expected payoff of player $i$ on the pure strategy profile $s$\\
			$\varpi^i$ & Sequence of actions taken by player $i$\\
			$\varpi^i_{I^j_i}$ & Sequence of player $i$ leading to $I^j_i$, $\varpi^i_h=\varpi^i_{I^j_i}$ for any $h\in I^j_i$ \\
			$\varpi^i_{I^j_i}a$ & The extended sequence $\varpi^i_{I^j_i}\cup \{a\}$\\
			${W}=\underset{i\in N_c}\prod{W}^i$ & The collection of sequence profiles, $\emptyset\in{W}^i$\\
			$g^i(\varpi)$ & Expected payoff of player $i$ on the sequence profile $\varpi$\\
			$\gamma^i$ & Realization plan of player $i\in N_c$\\
			$\Lambda=\underset{i\in N}\prod{ \Lambda^i}$ & Set of realization plan profiles\\
			$\Lambda_{++}=\underset{i\in N}\prod{ \Lambda^i_{++}}$ & Set of  strictly positive realization plan profiles\\
			$M_i(\varpi^i)$ & The index set of the information sets for player $i$ with $\varpi^i$ being the sequence\\
			$D_i$ & Set of $(j,a)$ for player $i$ with $M_i(\varpi^i_{I^j_i}a)=\emptyset$\\
			\bottomrule
		\end{tabular}
		\label{nfpr-tab-pre2}
	\end{table}
	The sequence form provides a compact alternative to the normal form by representing strategies through sequences rather than pure strategies. The sequence-form representation of $\Gamma$ is written as $\Gamma_s=\langle N,\{{W}^i\}_{i\in N_c},\gamma^c,\{g^i\}_{i\in N}\rangle$, and the relevant notation is summarized in Table~\ref{nfpr-tab-pre2}. For each player $i\in N_c$, a sequence $\varpi^i$ is the ordered collection of actions taken by player $i$ along some history. For any sequence profile $\varpi\in W$, the payoff function $g^i$ is defined by
	\begin{equation*}
		g^i(\varpi)=
		\begin{cases}
			u_z^i(h), & \text{if }\varpi\text{ is induced by some }h\in Z,\\
			0, & \text{otherwise.}
		\end{cases}
	\end{equation*}
	We say that a sequence profile $\varpi=(\varpi^i:i\in N_c)\in W$ is induced by a history $h\in H$ if $\varpi^i=\varpi_h^i$ for all $i\in N_c$. For each player $i\in N_c$, a realization plan is a function $\gamma^i$ on $W^i$ satisfying $\gamma^i(\varpi^i_\emptyset)=1$ and
	\begin{equation*}
		\begin{array}{l}
			\sum\limits_{a\in A(I_i^j)}\gamma^i(\varpi_{I_i^j}^i a)-\gamma^i(\varpi_{I_i^j}^i)=0,\; j\in M_i,\\
			0\le \gamma^i(\varpi_{I_i^j}^i a),\; j\in M_i,\ a\in A(I_i^j).
		\end{array}
	\end{equation*}
	Given a realization-plan profile $\gamma=(\gamma^i:i\in N_c)\in\Lambda$, the expected payoff of player $i\in N$ is $g^i(\gamma)=\sum_{\varpi^i\in W^i}\gamma^i(\varpi^i)\,g^i(\varpi^i,\gamma^{-i})$, where
	\[\textstyle
	g^i(\varpi^i,\gamma^{-i})
	=
	\sum\limits_{\varpi^{-i}\in W^{-i}}g^i(\varpi^i,\varpi^{-i})
	\prod\limits_{i_q\ne i}\gamma^{i_q}(\varpi^{i_q}).
	\]
	\begin{minipage}{0.40\textwidth}
		\centering
		\includegraphics[width=0.88\textwidth]{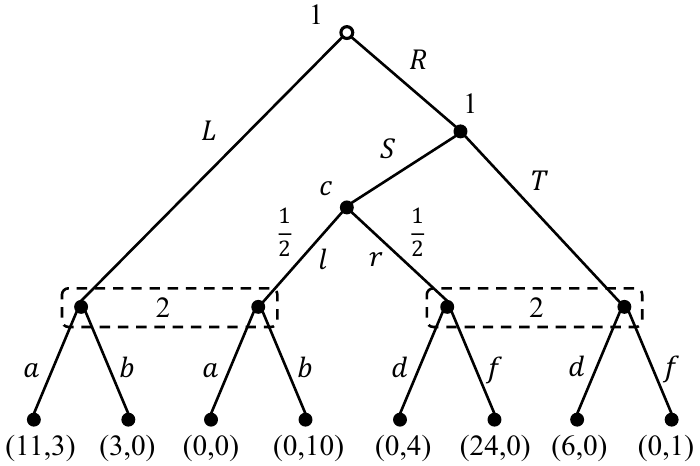}
		\captionof{figure}{An extensive-form game from von Stengel et al.~\cite{VonStengelComputingNormalForm2002}}
		\label{fig:game1}
	\end{minipage}\hfill
	\begin{minipage}{0.56\textwidth}
		\centering
		\captionof{table}{Reduced Normal Form of Fig.~\ref{fig:game1}}
		\label{ch2:tab:exm1nf}
		\fontsize{8.5pt}{11pt}\selectfont
		\begin{tabular}{lccc}
			\toprule
			& \multicolumn{3}{c}{Player 1} \\
			\cmidrule(lr){2-4}
			Player 2 
			& \(s_1^1=\{L\}\) 
			& \(s_2^1=\{R,S\}\) 
			& \(s_3^1=\{R,T\}\) \\
			\midrule
			\(s_1^2=\{a,d\}\) & \((11,3)\) & \((0,2)\)  & \((6,0)\) \\
			\(s_2^2=\{a,f\}\) & \((11,3)\) & \((12,0)\) & \((0,1)\) \\
			\(s_3^2=\{b,d\}\) & \((3,0)\)  & \((0,7)\)  & \((6,0)\) \\
			\(s_4^2=\{b,f\}\) & \((3,0)\)  & \((12,5)\) & \((0,1)\) \\
			\bottomrule
		\end{tabular}
	\end{minipage}
	\begin{example}\label{ch2:exm:extnote0} {\em
			Consider an extensive-form game $\Gamma$ shown in Fig.~\ref{fig:game1}, which is the game in Fig. 1 of von Stengel et al.~\cite{VonStengelComputingNormalForm2002}. The players' information sets are given by \( \mathcal{I}_1=\{I^1_1,I^2_1\} \), \( \mathcal{I}_2=\{I^1_2,I^2_2\} \), and \( \mathcal{I}_c=\{I^1_c\} \), where \( I^1_1=\{\emptyset\} \), \( I^2_1=\{\langle R\rangle\} \), \( I^1_2=\{\langle L\rangle,\langle R,S,l\rangle\} \), \( I^2_2=\{\langle R,S,r\rangle,\langle R,T\rangle\} \), and \( I^1_c=\{\langle R,S\rangle\} \). The pure strategies of the chance player are $s^c_1=\{l\},s^c_2=\{r\}$. The mixed strategy of the chance player is fixed, given by $\sigma^c=(\sigma^c(s^c_1),s^c_1(s^c_2))=(0.5,0.5)$. In the normal-form representation, the effect of chance can be incorporated directly into the payoff computation, thereby simplifying the analysis of pure-strategy profiles of the players. The normal-form representation of the extensive-form game can be summarized in Tab.~\ref{ch2:tab:exm1nf}. The corresponding mixed strategies are probability measures 
			$\sigma^1=(\sigma^1(s^1_1),\sigma^1(s^1_2),\sigma^1(s^1_3))^\top$, $\sigma^2=(\sigma^1(s^2_1),\sigma^1(s^2_2),\sigma^1(s^2_3),\sigma^1(s^2_4))^\top$. Based on Definition~\ref{nfpe-def-pre1} the Nash equilibria of the game can be derived manually. This game exhibits three distinct types of Nash equilibria, classified according to their final expected payoffs $u(\sigma)=(u^1(\sigma),u^2(\sigma))$.
			\begin{itemize}
				\item \textbf{Type A:} $\sigma^1 = (1,0,0)^\top$, $\sigma^2 = (\sigma^2(s^2_1), 1 - \sigma^2(s^2_1),0,0)^\top$ with $\frac{1}{12} \le \sigma^2(s^2_1) \le 1$; payoff $u(\sigma)=(11,3)$.
				\item \textbf{Type B:} $\sigma^1 = (0,\frac{1}{3},\frac{2}{3})^\top$, $\sigma^2 = (0,0,\frac{2}{3},\frac{1}{3})^\top$; payoff $u(\sigma)=(4,\frac{7}{3})$.
				\item \textbf{Type C:} $\sigma^1 = (\frac{5}{14}, \tfrac{3}{14}, \tfrac{3}{7})^\top$, $\sigma^2 = (\tfrac{1}{12}, \tfrac{1}{24},\tfrac{7}{12}, \tfrac{7}{24})^\top$; payoff $u(\sigma)=(4,\frac{3}{2})$.
			\end{itemize}
			The sequence-form representation of the extensive-form game can be summarized in Tab.~\ref{ch2:tab:exm1sf}.
			\begin{table}[!ht]
				\centering
				\caption{Sequence Form of Fig.~\ref{fig:game1}\label{ch2:tab:exm1sf}}
				\begin{tabular}{lccccc}
					\toprule
					& \multicolumn{5}{c}{Player 1 Sequences} \\
					\cmidrule(lr){2-6}
					Player 2 sequences
					& $\emptyset$ 
					& $\varpi^1_{I^1_1}L$ 
					& $\varpi^1_{I^1_1}R$ 
					& $\varpi^1_{I^2_1}S$ 
					& $\varpi^1_{I^2_1}T$ \\
					\midrule
					$\emptyset$           & (0,0) & (0,0)           & (0,0) & (0,0)           & (0,0) \\
					$\varpi^2_{I^1_2}a$   & (0,0) & \textbf{(11,3)} & (0,0) & \textbf{(0,0)}  & (0,0) \\
					$\varpi^2_{I^1_2}b$   & (0,0) & \textbf{(3,0)}  & (0,0) & \textbf{(0,5)}  & (0,0) \\
					$\varpi^2_{I^2_2}d$   & (0,0) & (0,0)           & (0,0) & \textbf{(0,2)}  & \textbf{(6,0)} \\
					$\varpi^2_{I^2_2}f$   & (0,0) & (0,0)           & (0,0) & \textbf{(12,0)} & \textbf{(0,1)} \\
					\bottomrule
				\end{tabular}
			\end{table}
		}
	\end{example}
 \section{Sequence-Form Nash Equilibrium and the Associated Polynomial System}\label{section3}
 Consider an extensive-form game $\Gamma$, with its normal form $\Gamma_n$ and its sequence form $\Gamma_s$. For any player $i\in N_c$ and pure strategy $s^i\in S^i$, define $s^i(\varpi^i)=\prod_{a\in \varpi^i}s^i(a), \, \varpi^i\in W^i$. For any mixed-strategy profile $\sigma\in\Xi$, let $\gamma(\sigma)=(\gamma^i(\sigma^i;\varpi^i): i\in N_c,\ \varpi^i\in W^i)$, where
 \begin{equation}\label{ms2rp}\textstyle
 	\gamma^i(\sigma^i;\varpi^i)=\sum\limits_{s^i\in S^i}s^i(\varpi^i)\sigma^i(s^i),
 	\; i\in N_c,\ \varpi^i\in W^i.
 \end{equation}
 In particular, for any pure strategy $s^i$, one has $\gamma^i(s^i;\varpi^i)=s^i(\varpi^i)$, and hence $\gamma(\sigma)\in\Lambda$. Now define $T=\{(\sigma,\gamma)\mid \sigma\in\Xi,\ \gamma=\gamma(\sigma)\}$. Then the following properties hold. For every $\gamma\in\Lambda$, there exists a mixed-strategy profile $\sigma$ such that $(\sigma,\gamma)\in T$. In particular, if $\gamma\in\Lambda_{++}$, one such profile is $\sigma(\gamma)=(\sigma^i(\gamma^i):i\in N_c)$, with $\sigma^i(\gamma^i;s^i)=\prod_{j\in M_i,\;a\in A(I_i^j),\;s^i(a)=1}\frac{\gamma^i(\varpi_{I_i^j}^i a)}{\gamma^i(\varpi_{I_i^j}^i)}$. Moreover, whenever $(\sigma,\gamma)\in T$, it holds that $u^i(\sigma)=g^i(\gamma)$ for every player $i\in N$.
 
 \subsection{Sequence-Form Nash Equilibrium}
 
Given $\gamma\in\Lambda$, we define the expected payoff, leading by the sequence $\varpi^i\in W^i$, for player $i\in N$ as 
 \[\begin{array}{l}
 	g^i(\gamma;\varpi^i)=\sum\limits_{\tilde{\varpi}^i\in W^i,\varpi^i\subseteq\tilde{\varpi}^i}\gamma^i(\tilde{\varpi}^i)g^i(\tilde{\varpi}^i,\gamma^{-i}),
 \end{array}\]
 and the maximal expected payoff attainable while committing to the sequence $\varpi^i\in W^i$ as
 \begin{equation}\label{ch3:maxpayoff}
 	\textstyle g^i_m(\varpi^i,\gamma^{-i})=\max\limits_{\tilde\gamma^i\in \Lambda^i,\tilde\gamma^i(\varpi^i)=1} g^i(\tilde\gamma^i,\gamma^{-i}).
 \end{equation}
 This notion of maximal expected payoff naturally leads to the definition of best-response sequences, which extends the standard concept of best responses to the sequence form.
 \begin{definition}\label{ch3:def:sfbr}{\em
 		Consider a realization plan profile $\gamma\in \Lambda$. For any $i\in N,j\in M_i,a\in A(I^j_i)$, we refer to $\varpi^i_{I^j_i}a$ as an $I^j_i$-best-response sequence to $\gamma$ if the following equality holds for any $a'\in A(I^j_i)$,
 		\begin{equation}\label{ch3:eqt:best}
 			\max\limits_{\tilde\gamma^i\in \Lambda^i} g^i(\tilde\gamma^i,\gamma^{-i};\varpi^i_{I^j_i}a)\geq\max\limits_{\tilde\gamma^i\in \Lambda^i} g^i(\tilde\gamma^i,\gamma^{-i};\varpi^i_{I^{j}_i}a').
 		\end{equation}
 		We define $\varpi^i_{I^j_i}a$ as a best-response sequence to $\gamma$ if, for any $j_q\in M_i,a_q\in A(I^{j_q}_i)$ with $a_q\in \varpi^i_{I^j_i}a$, $\varpi^i_{I^{j_q}_i}a_q$ qualifies as an $I^{j_q}_i$-best-response sequence to $\gamma$.}
 \end{definition}
 It is important to note that the maximal expected payoff $g^i_m(\varpi^i_{I^j_i}a,\gamma^{-i})$ can be decomposed into two components:
 \[
 \begin{array}{ll}
 	g^i_m(\varpi^i_{I^j_i}a,\gamma^{-i}) &= \max\limits_{\tilde\gamma^i\in \Lambda^i} g^i(\tilde\gamma^i,\gamma^{-i};\varpi^i_{I^j_i}a) \\
 	&\quad + \max\limits_{\tilde\gamma^i\in \Lambda^i, \tilde\gamma^i(\varpi^i_{I^j_i})=1} \sum\limits_{\tilde{\varpi}^i\in W^i, \tilde{\varpi}^i\cap A(I^j_i)=\emptyset} \gamma^i(\tilde{\varpi}^i) g^i(\tilde{\varpi}^i,\gamma^{-i}).
 \end{array}
 \]
 The second component involves only sequences that avoid the information set $I^j_i$, and thus is independent of the particular action $a$ chosen at $A(I^j_i)$. Therefore, Definition~\ref{ch3:def:sfbr} remains valid when the equation~(\ref{ch3:eqt:best}) is equivalently expressed in terms of maximal expected payoffs as
 \[
 g^i_m(\varpi^i_{I^j_i} a, \gamma^{-i}) \ge g^i_m(\varpi^i_{I^j_i} a', \gamma^{-i}).
 \]
 This reformulation highlights that best-response sequences can be characterized entirely through maximal expected payoffs, providing a more concise and convenient representation for subsequent derivations. We now establish that a sequence is a best-response sequence if and only if it attains the maximal expected payoff $g^i_m$ among all sequences. 
 
 Analogous to the standard notion of a pure-strategy best response, a best-response sequence can be equivalently reformulated in terms of its maximal expected payoff across all sequences, as stated in the following proposition.
 
 \begin{theorem}{\em
 		For player $i\in N$, $\varpi^i\in W^i$ is a best-response sequence to a given $\gamma\in\Lambda$ if and only if $g^i_m(\varpi^i,\gamma^{-i})\geq g^i_m(\tilde \varpi^i,\gamma^{-i})$ for any $\tilde \varpi^i\in W^i$.}
 \end{theorem}
 \begin{proof}
 	The sufficiency direction is immediate from the definition. We now show necessity. For any $i\in N,j\in M_i,a\in A(I^j_i)$, suppose that $\varpi^i_{I^j_i}a$ is an $I^j_i$-best-response sequence to $\gamma$, then $g^i_m(\varpi^i_{I^j_i}a,\gamma^{-i})= g^i_m(\varpi^i_{I^j_i},\gamma^{-i})$. Next, consider any $\varpi^i\in W^i$ that is a best-response sequence to a given $\gamma\in\Lambda$. By recursively invoking the local equality above along the sequence, one arrives at $g^i_m(\varpi^i,\gamma^{-i})= g^i_m(\emptyset,\gamma^{-i})$. This identity immediately yields $g^i_m(\varpi^i,\gamma^{-i})\geq g^i_m(\tilde \varpi^i,\gamma^{-i})$ for any $\tilde \varpi^i\in W^i$, completing the proof.
 \end{proof}
 This proposition naturally leads to the definition of a sequence-form Nash equilibrium.
 \begin{definition}\label{ch3:def:sfne}{\em
 		Let $\Gamma$ be an extensive form game.  A realization plan profile $\gamma^*\in \Lambda$ is a sequence-form Nash equilibrium if, for any player $i\in N$, $\gamma^{*i}(\varpi^i)=0$ whenever $g^i_m(\varpi^i,\gamma^{*-i})<g^i_m(\tilde \varpi^i,\gamma^{*-i})$ for some $\tilde \varpi^i\in W^i$.}
 \end{definition}

 In order to prove the equivalence between mixed-strategy Nash equilibria and sequence-form Nash equilibria, we first introduce a lemma from Hou et al.~\cite{houSequenceFormCharacterizationDifferentiable2025}. This lemma captures the relationship between best-response sequences and the optimal pure strategies.
 \begin{lemma}\label{ch3:lem:bestres}{\em
 		For $(\sigma,\gamma)\in T$ and player $i$, the following statements are equivalent:
 		\begin{enumerate}
 			\item $u^i(s^i,\sigma^{-i})\geq u^i(\tilde s^i,\sigma^{-i})$ holds for any $\tilde s^i\in S^i$.
 			\item For any $j\in M_i,a\in A(I^j_i)$ with $s^i(\varpi^i_{I^j_i}a)=1$, $\varpi^i_{I^j_i}a$ is a best-response sequence to $\gamma$.
 	\end{enumerate} }
 \end{lemma}

 \begin{theorem}\label{ch3:thm:neeq}{\em
 		For $(\sigma^*,\gamma^*)\in T$, $\sigma^*$ is a Nash equilibrium if and only if $\gamma^*$ is a sequence-form Nash equilibrium.}
 \end{theorem}
 \begin{proof}
 	\textbf{(Sufficiency)}: Suppose that $\gamma^*$ is a sequence-form Nash equilibrium, then any sequence $\varpi^i\in W^i$ of play $i$ with $\gamma^{*i}(\varpi^i)>0$ is a best-response sequence to $\gamma^*$. Consider $s^i\in S^i$ such that $\sigma^{*i}(s^i)>0$. Then, for any $\varpi^i$ with $s^i(\varpi^i) = 1$, we have $\gamma^{*i}(\varpi^i)=\sum_{\tilde s^i\in S^i}\tilde s^i(\varpi^i)\sigma^{*i}(\tilde s^i)\geq s^i(\varpi^i)\sigma^{*i}(s^i)>0$. implying that $\varpi^i$ is a best-response sequence to $\gamma^*$. By Lemma~\ref{ch3:lem:bestres}, it follows that $u^i(s^i, \sigma^{*-i}) \ge u^i(\tilde{s}^i, \sigma^{*-i})$ holds for any $\tilde{s}^i \in S^i$, and therefore $\sigma^*$ is a Nash equilibrium.
 	
 	\textbf{(Necessity)}: Assume that $\sigma^*$ is a Nash equilibrium. We show that any sequence $\varpi^i\in W^i$ satisfying $\gamma^{*i}(\varpi^i)>0$ is a best-response sequence to $\gamma^*$. Since $\gamma^{*i}(\varpi^i)=\sum_{s^i\in S^i}s^i(\varpi^i)\sigma^{*i}(s^i)>0$, there exists $s^i\in S^i$ such that $s^i(\varpi^i)=1$ and $\sigma^{*i}(s^i)>0$. Consequently, $u^i(s^i,\sigma^{*-i})\geq u^i(\tilde s^i,\sigma^{*-i})$ for all $\tilde s^i\in S^i$. It then follows from Lemma~\ref{ch3:lem:bestres} that $\varpi^i$ is a best-response sequence to $\gamma^*$. Thus $\gamma^*$ is a sequence-form Nash equilibrium.
 \end{proof}
 \subsection{A Polynomial Equilibrium System}
 To make sequence-form Nash equilibrium computationally tractable, we derive an equivalent system that formally captures the equilibrium conditions. For any fixed profile of the opponents' realization plans $\hat\gamma^{-i}$, player $i$'s best response can be obtained by solving the following linear optimization problem,
 \begin{equation}\label{ch3:opt:ne}
 	\begin{array}{rl}
 		\max\limits_{\gamma^i} &\sum\limits_{j\in M_i}\sum\limits_{a\in A(I^j_i)}\gamma^i(\varpi^i_{I^j_i}a)g^i(\varpi^i_{I^j_i}a,\hat\gamma^{-i})\\
 		\text{s.t.}&\sum\limits_{a\in A(I^j_i)}\gamma^i(\varpi^i_{I^j_i}a)-\gamma^i(\varpi^i_{I^j_i})=0,\;j\in M_i,\\
 		&0\le \gamma^i(\varpi^i_{I^j_i}a),\; (j,a)\in D_i.
 	\end{array}
 \end{equation}
 Two clarifications regarding Problem~(\ref{ch3:opt:ne}) are necessary. Firstly, the payoff associated with the empty sequence is omitted, since it remains constant with respect to $\gamma^i$ and therefore does not affect optimality. Secondly, we exclude redundant inequalities in the constraints that arise from the recursive equation. 
 
 By applying the Karush–Kuhn–Tucker (KKT) optimality conditions to Problem~(\ref{ch3:opt:ne}) for each player and subsequently imposing the consistency condition $\gamma=\hat\gamma$, we obtain the following theorem.
 \begin{theorem}\label{ch3:thm:ne} {\em $\gamma^*\in\Lambda$ constitutes a sequence-form Nash equilibrium if and only if there exists a pair $(\lambda^*,\nu^*)$ alongside $\gamma^*$ that satisfies the system,
 		\begin{equation}
 			\label{ch3:eqt:ne}
 			\begin{array}{l}
 				g^i(\varpi^i_{I^j_i}a,\gamma^{-i})+ \lambda^i(\varpi^i_{I^j_i}a)-\nu^i_{I^j_i} = 0,\;i\in N,(j,a)\in D_i,\\
 				g^i(\varpi^i_{I^j_i}a,\gamma^{-i}) -\nu^i_{I^j_i} + \zeta^i_{I^j_i}(a) = 0,\;i\in N,(j,a)\notin D_i,\\
 				\sum\limits_{a\in A(I^j_i)}\gamma^i(\varpi^i_{I^j_i}a)-\gamma^i(\varpi^i_{I^j_i})=0,\;i\in N,j\in M_i,\\
 				\gamma^i(\varpi^i_{I^j_i}a)\lambda^i(\varpi^i_{I^j_i}a)=0,\;0<\gamma^i(\varpi^i_{I^j_i}a),\;0<\lambda^i(\varpi^i_{I^j_i}a),\; i\in N,(j,a)\in D_i,
 			\end{array}
 		\end{equation}
 		where $\zeta^i_{I^j_i}(a)=\sum_{{j_q}\in M_i(\varpi^i_{I^j_i}a)}\nu^i_{I^{j_q}_i}$.
 	}
 \end{theorem}
 \begin{proof} \textbf{(Sufficiency)}: Assume that there exists a pair $(\lambda^*,\nu^*)$ satisfying System~\eqref{ch3:eqt:ne} together with $\gamma^*$. We show that for any $i\in N, (j,a)\in D_i$ with $\gamma^{*i}(\varpi^i_{I^j_i}a)>0$, $\varpi^i_{I^j_i}a$ is a best-response sequence to $\gamma^*$. Suppose, to the contrary, that $\varpi^i_{I^j_i}a$ is not a best-response sequence to $\gamma^*$. Then there exists some $j_q\in M_i$ and $a_q\in A(I^{j_q}_i)$ with $a_q\in \varpi^i_{I^j_i}a$ such that $\varpi^i_{I^{j_q}_i}a_q$ fails to be an $I^{j_q}_i$-best-response sequence to $\gamma^*$. Consequently, there exists an alternative action $a'\in A(I^{j_q}_i)$ satisfying
 	\begin{equation}\label{ch3:eqt:thmpoly}\begin{array}{l}
 			\max\limits_{\gamma^i\in \Lambda^i} g^i(\gamma^i,\gamma^{*-i};\varpi^i_{I^{j_q}_i}a_q)<\max\limits_{\gamma^i\in \Lambda^i} g^i(\gamma^i,\gamma^{*-i};\varpi^i_{I^{j_q}_i}a').
 	\end{array}\end{equation}
 	Let $\tilde\gamma^i\in \arg\max_{\gamma^i\in \Lambda^i} g^i(\gamma^i,\gamma^{*-i};\varpi^i_{I^{j_q}_i}a')$ and define $\tilde\gamma^{*i}$ such that $\tilde\gamma^{*i}(\varpi^i)=\tilde\gamma^{i}(\varpi^i)$ for all $\varpi^i\in W^i$ with $\varpi^i_{I^{j_q}_i}\nsubseteq\varpi^i$, while $\tilde\gamma^{*i}(\varpi^i)=\gamma^{*i}(\varpi^i)/\gamma^{*i}(\varpi^i_{I^j_i}a)$ whenever $\varpi^i_{I^{j_q}_i}\subseteq\varpi^i$. It follows from (\ref{ch3:eqt:thmpoly}) that $g^i(\tilde\gamma^i,\gamma^{*-i})>g^i(\tilde\gamma^{*i},\gamma^{*-i})$. For sufficiently small $\epsilon>0$, define $\gamma^i=\gamma^{*i}+\epsilon(\tilde\gamma^i-\tilde\gamma^{*i})$. Then $\gamma^i\in\Lambda$ and $g^i(\gamma^i,\gamma^{*-i})>g^i(\gamma^{*i},\gamma^{*-i})$, which contradicts the optimality of $\gamma^{*i}$ in Problem~(\ref{ch3:opt:ne}). Therefore, the claim follows.
 	
 	\textbf{(Necessity)}: Conversely, suppose that $\gamma^*$ is a sequence-form Nash equilibrium. Then, for any player $i\in N$, we have
 	\[\begin{array}{ll}
 		\max\limits_{\gamma^i\in \Lambda^i} g^i(\gamma^i,\gamma^{*-i})&=\sum\limits_{j\in M_i(\emptyset)}\max\limits_{\gamma^i\in \Lambda^i}\sum\limits_{a\in A(I^j_i)}\gamma^i(\varpi^i_{I^j_i}a)g^i(\gamma^i,\gamma^{*-i};\varpi^i_{I^j_i}a)+g^i(\emptyset,\gamma^{*-i})\\
 		&=\sum\limits_{j\in M_i(\emptyset)}\sum\limits_{a\in A(I^j_i)}\gamma^{*i}(\varpi^i_{I^j_i}a)\max\limits_{\gamma^i\in \Lambda^i} g^i(\gamma^i,\gamma^{*-i};\varpi^i_{I^j_i}a)+g^i(\emptyset,\gamma^{*-i})\\
 		&=\sum\limits_{\varpi^i\in W^i} \gamma^{*i}(\varpi^i)g^i(\varpi^i,\gamma^{*-i})\\
 		&=g^i(\gamma^{*i},\gamma^{*-i}).
 	\end{array}
 	\]
 	Hence, $\gamma^{*i}$ attains the maximal payoff against $\gamma^{*-i}$, which completes the proof.
 \end{proof}
 \section{Interior-Point Differentiable Path-Following Methods}
 \label{section4}
 This section utilizes the sequence-form necessary and sufficient condition in Theorem~\ref{ch3:thm:ne} to develop two interior-point differentiable path-following methods for finding Nash equilibria. These methods generate a smooth path to a Nash equilibrium by constructing artificial games whose equilibria vary continuously with an extra variable. Specifically, these two methods introduce, respectively, different logarithmic barrier terms into the payoff function. These barrier terms ensure that the strategy variables remain in the interior of the feasible region, thereby enhancing the smoothness and numerical stability of the path.
 
 \subsection{A Logarithmic-Barrier Smooth Path}
 \label{ch3:sec:logit}
 Let $\gamma^0=(\gamma^{0i}(\varpi^i):i\in N,\varpi^i\in W^i)$ denote a given realization plan profile with $\gamma^{0i}(\varpi^i)>0$, which serves as a starting point for the smooth path discussed later. For $t\in(0,1]$, we constitute a logarithmic-barrier artificial game $\Gamma_{s}^{l_n}(t)$ in the sequence form where each player $i$ derives their optimal response to a given strategy $\hat\gamma\in\Lambda$ by solving the strictly convex optimization problem,
 \begin{equation}
 	\label{ch3:opt:lgne}
 	\begin{array}{rl}
 		\max\limits_{\gamma^i} & (1-t)\sum\limits_{j\in M_i}\sum\limits_{a\in A(I^j_i)}\gamma^i(\varpi^i_{I^j_i}a)g^i(\varpi^i_{I^j_i}a,\hat\gamma^{-i})\\
 		&+t\sum\limits_{(j,a)\in D_i}(\gamma^{0i}(\varpi^i_{I^j_i}a)\ln\gamma^i(\varpi^i_{I^j_i}a)-\gamma^i(\varpi^i_{I^j_i}a))\\
 		\text{s.t.} & \sum\limits_{a\in A(I^j_i)}\gamma^i(\varpi^i_{I^j_i}a)-\gamma^i(\varpi^i_{I^j_i})=0,\;j\in M_i.
 	\end{array}
 \end{equation}
 In this objective function, the first term represents the ultimate optimization goal. The second term constrains the strategy space to be strictly positive and enables the explicit design of the starting point. According to the Nash equilibrium principle, $\gamma^*$ is said to be a Nash equilibrium of the artificial game $\Gamma_{s}^{l_n}(t)$ if and only if, for every player $i\in N$, the strategy $\gamma^{*i}$ is an optimal solution to Problem~(\ref{ch3:opt:lgne}) given the strategies $\gamma^{*-i}$ of the other players.
 
 By imposing the KKT optimality conditions on Problem~(\ref{ch3:opt:lgne}) simultaneously for all players and setting $\hat\gamma=\gamma$, we obtain the following polynomial equilibrium system of $\Gamma_{s}^{l_n}(t)$,
 \begin{equation}
 	\label{ch3:eqt:lgne}
 	\begin{array}{l}
 		(1-t)g^i(\varpi^i_{I^j_i}a,\gamma^{-i})+\lambda^i(\varpi^i_{I^j_i}a)-t-\nu^i_{I^j_i} = 0,\;i\in N,(j,a)\in D_i,\\
 		(1-t)g^i(\varpi^i_{I^j_i}a,\gamma^{-i})-\nu^i_{I^j_i} + \zeta^i_{I^j_i}(a) = 0,\;i\in N,(j,a)\notin D_i,\\
 		\sum\limits_{a\in A(I^j_i)}\gamma^i(\varpi^i_{I^j_i}a)-\gamma^i(\varpi^i_{I^j_i})=0,\;i\in N,j\in M_i,\\
 		\gamma^i(\varpi^i_{I^j_i}a)\lambda^i(\varpi^i_{I^j_i}a)=t\gamma^{0i}(\varpi^i_{I^j_i}a),\;0<\gamma^i(\varpi^i_{I^j_i}a),\;0<\lambda^i(\varpi^i_{I^j_i}a),\;i\in N,(j,a)\in D_i,
 	\end{array}
 \end{equation}
 where $\zeta^i_{I^j_i}(a)=\sum_{{j_q}\in M_i(\varpi^i_{I^j_i}a)}\nu^i_{I^{j_q}_i}$. 
 
 One can see that $\gamma^*$ constitutes a Nash equilibrium of the artificial game $\Gamma_{s}^{l_n}(t)$ if and only if there exists a unique pair of $(\lambda^*,\nu^*)$ together with $\gamma^*$ satisfying System~(\ref{ch3:eqt:lgne}). Since $\gamma^* \in \Lambda_{++}$, the mixed-strategy profile $\sigma(\gamma^*)$ is well-defined and satisfy $(\sigma(\gamma^*), \gamma^*) \in T$.
 
 \begin{theorem}\label{ch3:thm:lgnelim}{\em
 		Let $\{\gamma^*(t_k)\}_{k=1}^\infty$ be a sequence of realization-plan profiles generated by System~(\ref{ch3:eqt:lgne}) with $t=t_k\in(0,1]$ and $\lim\limits_{k\to \infty}t_k=0$. Then every limit point of the sequence $\{\gamma^*(t_k)\}_{k=1}^\infty$ yields a sequence-form Nash equilibrium.
 	}
 \end{theorem}
 \begin{proof}
 	Since each $\gamma^*(t_k)$ lies in the bounded set $\Lambda_{++}$, the sequence $\{\gamma^*(t_k)\}_{k=1}^\infty$ admits a convergent subsequence. For simplicity, we denote this subsequence by the same symbol. Let $\gamma^*=\lim_{k\to\infty}\gamma^*(t_k)$, establishing that $\gamma^*\in \Lambda$. Noting that the equations in System~(\ref{ch3:eqt:lgne}) reduce to those of System~(\ref{ch3:eqt:ne}) when $t = 0$, it follows that $\gamma^*$ serves as a solution component for System (\ref{ch3:eqt:ne}). Consequently, $\gamma^*$ is a sequence-form Nash equilibrium. 
 \end{proof}
 It immediately follows that any limit point of the sequence $\{\sigma(\gamma^*(t_k))\}^\infty_{k=1}$ yields a mixed-strategy Nash equilibrium.
 
 Theorem~\ref{ch3:thm:lgnelim} establishes that as the auxiliary variable $t$ approaches zero, the sequence of realization-plan profiles generated by System~(\ref{ch3:eqt:lgne}) converges to a Nash equilibrium. Next, we will establish the existence of a smooth path along which the points satisfy System~(\ref{ch3:eqt:lgne}). This path initiates from a totally mixed strategy profile and ultimately derives a Nash equilibrium.
 
 \begin{lemma}\label{ch3:thm:lgnests} {\em At $t=1$, System~(\ref{ch3:eqt:lgne}) has a unique solution given by $(\gamma^*(1),\lambda^*(1),\nu^*(1))$ with $\gamma^{*i}(1;\varpi^i_{I^j_i}a)=\gamma^{0i}(\varpi^i_{I^j_i}a)$, $\lambda^{*i}(1;\varpi^{i}_{I^j_i}a)=1$ and $\nu^{*i}_{I^j_i}(1)=1$.}
 \end{lemma}
 \begin{proof}
 	At $t=1$, Problem (\ref{ch3:opt:lgne}) is strictly convex and therefore admits at most one optimal solution. Applying KKT optimality conditions for each player yields the following system,
 	\begin{equation}\label{ch3:eqt:lgnests}\begin{array}{l}
 			\lambda^i(\varpi^i_{I^j_i}a)-1-\nu^i_{I^j_i} = 0,\;i\in N,(j,a)\in D_i,\\
 			-\nu^i_{I^j_i} + \zeta^i_{I^j_i}(a) = 0,\;i\in N,(j,a)\notin D_i,\\
 			\sum\limits_{a\in A(I^j_i)}\gamma^i(\varpi^i_{I^j_i}a)-\gamma^i(\varpi^i_{I^j_i})=0,\;i\in N,j\in M_i,\\
 			\gamma^i(\varpi^i_{I^j_i}a)\lambda^i(\varpi^i_{I^j_i}a)=\gamma^{0i}(\varpi^i_{I^j_i}a),\;0<\gamma^i(\varpi^i_{I^j_i}a),\;0<\lambda^i(\varpi^i_{I^j_i}a),\;i\in N,(j,a)\in D_i.\\
 		\end{array}
 	\end{equation}
 	This system coincides with~(\ref{ch3:eqt:lgne}) evaluated at $t=1$, since the objective function becomes independent of $\hat{\gamma}$. We consider the candidate solution $\gamma^*(1)$ defined by $\gamma^{*i}(1;\varpi^i_{I^j_i}a)=\gamma^{0i}(\varpi^i_{I^j_i}a)$. Substituting $\gamma=\gamma^*(1)$ to System (\ref{ch3:eqt:lgnests}) results in a unique solution for $(\lambda,\nu)$, denoted by $(\lambda^*(1),\nu^*(1))$ with $\lambda^{*i}(1;\varpi^i_{I^j_i}a)=1$ and $\nu^{*i}_{I^j_i}(1)=1$. As a result, $\gamma^{*i}(1)$ is the unique optimal solution to (\ref{ch3:opt:lgne}) and System~(\ref{ch3:eqt:lgne}) admits a unique solution $(\gamma^*(1),\lambda^*(1),\nu^*(1))$ at $t=1$. This completes the proof.
 \end{proof}
 
 Lemma~\ref{ch3:thm:lgnests} shows that System~(\ref{ch3:eqt:lgne}) possesses a unique solution at $t=1$. In the subsequent discussion, we demonstrate the existence of a connected component that intersects both the $t=1$ and $t=0$ levels. Before progressing further, it is essential to introduce Browder's fixed-point theorem~\cite{Browdercontinuityfixedpoints1960}.
 \begin{theorem}\label{ch3:thm:Browder}{\em\textbf{(Browder's fixed point theorem).} Let $C$ be a nonempty, compact and convex subset of $\mathbb{R}^m$ and $h:C\times[0,1]\to C$ be a continuous function. Then the set $H =\{(z,t)\in C\times[0,1]|z=h(z,t)\}$ contains a connected set $H^c$ such that $C\times\{1\}\cap H^c\neq\emptyset$ and $C\times\{0\}\cap H^c\neq\emptyset$.}
 \end{theorem}
 Let $\widetilde{\mathscr{S}}_L=\{(\gamma,\lambda,\nu,t)|(\gamma,\lambda,\nu,t) \text{ satisfies System~(\ref{ch3:eqt:lgne}) with } 0<t\leq 1\}$ and $\mathscr{S}_L$ be the closure of $\widetilde{\mathscr{S}}_L$. By applying Theorem~\ref{ch3:thm:Browder}, we arrive at the following conclusion.
 \begin{theorem}\label{ch3:thm:lgnecnt}{\em There is a connected component in $\mathscr{S}_L$ intersecting both $\mathbb{R}^{n_0}\times\mathbb{R}^{n_0}\times\mathbb{R}^{m_0}\times\{1\}$ and $\mathbb{R}^{n_0}\times\mathbb{R}^{n_0}\times\mathbb{R}^{m_0}\times\{0\}$.}
 \end{theorem}
 \begin{proof}
 	For $(\hat{\gamma},t)\in \Lambda\times[0,1]$, define $\varphi(\hat{\gamma},t)=(\gamma^i:i\in N)$, where $\gamma^i$ denotes the unique solution to the strictly convex optimization problem,
 	\begin{equation}
 		\label{ch3:eqt:conect}
 		\begin{array}{rl}
 			\max\limits_{\gamma^i} & (1-t)\sum\limits_{j\in M_i}\sum\limits_{a\in A(I^j_i)}\gamma^i(\varpi^i_{I^j_i}a)g^i(\varpi^i_{I^j_i}a,\hat\gamma^{-i})\\
 			&+t\sum\limits_{(j,a)\in D_i}(\gamma^{0i}(\varpi^i_{I^j_i}a)\ln\gamma^i(\varpi^i_{I^j_i}a)-\gamma^i(\varpi^i_{I^j_i}a))\\
 			&-\frac{1}{2}\sum\limits_{(j,a)\in D_i}(\gamma^i(\varpi^i_{I^j_i}a)-\hat\gamma(\varpi^i_{I^j_i}a))^2\\
 			\text{s.t.} & \sum\limits_{a\in A(I^j_i)}\gamma^i(\varpi^i_{I^j_i}a)-\gamma^i(\varpi^i_{I^j_i})=0,\;j\in M_i.
 		\end{array}
 	\end{equation}
 	Based on Theorem 2.2.2 in~\cite{FiaccoIntroductionSensitivityStability1983}, it follows that $\varphi(\gamma, t)$ is a continuous function that maps from $\Lambda \times [0,1]$ to $\Lambda$. Let $\mathscr{E}=\{(\gamma,t)\in\Lambda\times[0,1]|\varphi(\gamma,t)=\gamma\}$. Theorem~\ref{ch3:thm:Browder} ensures the existence of a connected component within $\mathscr{E}$ that intersects both $\mathbb{R}^{n_0}\times\{1\}$ and $\mathbb{R}^{n_0}\times\{0\}$. We denote this connected component as $\mathscr{E}^c$, and specifically refer to the portion where $t>0$ as $\widetilde{\mathscr{E}}^c$.
 	
 	By employing the optimality condition to Problem~(\ref{ch3:eqt:conect}), we derive a polynomial system that coincides with System~(\ref{ch3:eqt:lgne}). Hence, for any $(\gamma,t)\in \widetilde{\mathscr{E}}^c$, there exists a unique pair $(\lambda,\nu)$ such that System (\ref{ch3:eqt:lgne}) is satisfied. Let $\widetilde{\mathscr{S}}^c_L=\{(\gamma,\lambda,\nu,t)\in \widetilde{\mathscr{S}}_L|(\gamma,t)\in \widetilde{\mathscr{E}}^c\}$ and $\mathscr{S}^c_L$ be the closure of $\widetilde{\mathscr{S}}^c_L$. We obtain from the above discussion that $\mathscr{S}^c_L$ constitutes a connected component within $\mathscr{S}_L$ intersecting $\mathbb{R}^{n_0}\times\mathbb{R}^{n_0}\times\mathbb{R}^{m_0}\times\{1\}$. Considering a convergent sequence $\{(\gamma(t_k), t_k), k = 1, 2, \ldots\} \subseteq \widetilde{\mathscr{E}}^c$ with $\lim_{k\to\infty} t_k = 0$, we associate each $(\gamma(t_k), t_k)$ with the corresponding pair $(\lambda(t_k), \mu(t_k))$, which is bounded as shown in Appendix~\ref{app:compactness_sl}. The boundedness of $\{(\gamma(t_k),\lambda(t_k),\mu(t_k),t_k),k=1,2,\ldots\}\subseteq \widetilde{\mathscr{S}}^c_L$ guarantees that it has a convergent subsequence. Thus, $\mathscr{S}^c_L$ intersects with $\mathbb{R}^{n_0}\times\mathbb{R}^{n_0}\times\mathbb{R}^{m_0}\times\{0\}$. This completes the proof.
 \end{proof}
 According to Lemma~\ref{ch3:thm:lgnests}, the connected component described in Theorem~\ref{ch3:thm:lgnecnt} is unique and intersects the $t=1$ level at $(\gamma^*(1),\lambda^*(1),\nu^*(1),1)$. Let $\alpha=(\alpha(\varpi^i_{I^j_i}a):i\in N, j\in M_i, a\in A(I^j_i))\in\mathbb{R}^{n_0}$ be an arbitrary vector with sufficiently small $\|\alpha\|$. To realize a smooth path, System~(\ref{ch3:eqt:lgne}) is accordingly modified. Specifically, we subtract the expression $t(1-t)\alpha$ from the left-hand side of the first group of equations, resulting in a new system. Let $p(\gamma,\lambda,\nu,t;\alpha)$ represent the left-hand sides of the equations in the newly obtained system. When treating $\alpha$ as a constant, we define $p_\alpha(\gamma,\lambda,\nu,t) = p(\gamma,\lambda,\nu,t;\alpha)$. This gives rise to the following theorem.
 \begin{theorem}\label{ch3:thm:lgnesmp}{\em Given almost any $\alpha\in\mathbb{R}^{n_0}$ with sufficiently small $\|\alpha\|$, there exists a smooth path in $\mathscr{S}_L$ that starts from $(\gamma^*(1),\lambda^*(1),\nu^*(1),1)$ on the level of $t=1$ and leads to a sequence-form Nash equilibrium as $t$ approaches zero.}
 \end{theorem}
 \begin{proof}
 	The second group of equations and inequality constraints in System~(\ref{ch3:eqt:lgne}) reveals that the elements in $\mathscr{S}_L$ satisfy $\gamma\in\Lambda_{++}$ for $t\in(0,1)$. With the continuous differentiability of $p(\gamma,\lambda,\nu,t;\alpha)$ on $\Lambda_{++}\times\mathbb{R}^{n_0}_{++}\times\mathbb{R}^{m_0}\times(0,1)\times\mathbb{R}^{n_0}$, we have proved in Appendix~\ref{app:fullrowrank} that the Jacobian matrix of $p(\gamma,\lambda,\nu,t;\alpha)$ is of full-row rank in this region. As an application of the transversality theorem outlined by Eaves and Schmedders~\cite{EavesGeneralequilibriummodels1999}, it can be shown that zero is a regular value of $p_\alpha(\gamma,\lambda,\nu,t)$ over $\Lambda_{++}\times\mathbb{R}^{n_0}_{++}\times\mathbb{R}^{m_0}\times(0,1)$ for almost any $\alpha$.
 	
 	We fix $\alpha$ such that zero is a regular value of $p_\alpha(\gamma,\lambda,\nu,t)$ over $\Lambda_{++}\times\mathbb{R}^{n_0}_{++}\times\mathbb{R}^{m_0}\times(0,1)$. By applying the implicit function theorem, the component described in Theorem~\ref{ch3:thm:lgnecnt} defines a smooth path, originating at $(\gamma^*(1),\lambda^*(1),\nu^*(1),1)$ when $t=1$ and terminates at $t=0$. In Appendix~\ref{app:fullrowrank}, we demonstrate that, at $t=1$, zero remains a regular value of $p_\alpha(\gamma,\lambda,\nu,1)$ in $\Lambda_{++}\times\mathbb{R}^{n_0}_{++}\times\mathbb{R}^{m_0}$. This implies that the smooth path does not intersect tangentially with $\mathbb{R}^{n_0}\times\mathbb{R}^{n_0}\times\mathbb{R}^{m_0}\times\{1\}$. Theorems~\ref{ch3:thm:lgnelim} and \ref{ch3:thm:lgnecnt} show that this smooth path ultimately yields a sequence-form Nash equilibrium. This completes the proof.
 \end{proof}
 To facilitate more efficient numerical computations, we present an equivalent reformulation based on Cao et al.~\cite{Caodifferentiablepathfollowingmethod2022}. This reformulation is able to significantly reduce the number of variables involved.
 
 Given $\tau_0>0$ and $\kappa_0>2$, define
 \begin{small} \[\psi_1(v,r;\tau_0,\kappa_0)=\left(\frac{v+\sqrt{v^2+4\tau_0r}}{2}\right)^{\kappa_0} \text{ and } \psi_2(v,r;\tau_0,\kappa_0)=\left(\frac{-v+\sqrt{v^2+4\tau_0r}}{2}\right)^{\kappa_0}.\]\end{small}It follows that $\psi_1(v,r;\tau_0,\kappa_0)\psi_2(v,r;\tau_0,\kappa_0)=(\tau_0r)^{\kappa_0}$. Since $\kappa_0>2$, $\psi_1(v,r;\tau_0,\kappa_0)$ and $\psi_2(v,r;\tau_0,\kappa_0)$ are both continuously differentiable on $\mathbb{R}\times[0,\infty)$. For $x=(x^i(\varpi^i_{I^j_i}a):i\in N,j\in M_i,a\in A(I^j_i))\in \mathbb{R}^{n_0}$, consider the functions $\gamma(x,t)=(\gamma^i(x,t;\varpi^i_{I^j_i}a):i\in N,j\in M_i,a\in A(I^j_i))$ and $\lambda(x,t)=(\lambda^i(x,t;\varpi^i_{I^j_i}a):i\in N,j\in M_i,a\in A(I^j_i))$ where,
 \begin{equation}\label{ch3:eqt:lgnesubs}\begin{array}{l}
 		\gamma^i(x,t;\varpi^i_{I^j_i}a)=\psi_1(x^i(\varpi^i_{I^j_i}a),t^{1/\kappa_0}; \gamma^{0i}(\varpi^i_{I^j_i}a)^{1/\kappa_0}, \kappa_0),\\
 		\lambda^i(x,t;\varpi^i_{I^j_i}a)=\psi_2(x^i(\varpi^i_{I^j_i}a),t^{1/\kappa_0};\gamma^{0i}(\varpi^i_{I^j_i}a)^{1/\kappa_0}, \kappa_0),\; i\in N,(j,a)\in D_i,\\
 		\gamma^i(x,t;\varpi^i_{I^j_i}a)=x^i(\varpi^i_{I^j_i}a),\; i\in N,(j,a)\notin D_i.
 \end{array}\end{equation}
 It is evident that $\gamma^i(x,t;\varpi^i_{I^j_i}a)\lambda^i(x,t;\varpi^i_{I^j_i}a)=t\gamma^{0i}(\varpi^i_{I^j_i}a)$ for $i\in N,(j,a)\in D_i$. By substituting $\gamma^i(x,t;\varpi^i_{I^j_i}a)$ and $\lambda^i(x,t;\varpi^i_{I^j_i}a)$ into the system $p(\gamma,\lambda,\nu,t;\alpha)=0$ for $\gamma^i(\varpi^i_{I^j_i}a)$ and $\lambda^i(\varpi^i_{I^j_i}a)$, we obtain an equivalent formulation with fewer variables and constraints,
 \begin{equation}\label{ch3:eqt:lgnetrans}\begin{array}{l}
 		(1-t)g^i(\varpi^i_{I^j_i}a,\gamma^{-i}(x,t))+\lambda^i(x,t;\varpi^i_{I^j_i}a)\\
 		\hspace{2.5cm}-t-\nu^i_{I^j_i} -t(1-t)\alpha(\varpi^i_{I^j_i}a)= 0,\; i\in N,(j,a)\in D_i,\\
 		(1-t)g^i(\varpi^i_{I^j_i}a,\gamma^{-i}(x,t))-\nu^i_{I^j_i} + \zeta^i_{I^j_i}(a)-t(1-t)\alpha(\varpi^i_{I^j_i}a) = 0,\; i\in N,(j,a)\notin D_i,\\
 		\sum\limits_{a\in A(I^j_i)}\gamma^i(x,t;\varpi^i_{I^j_i}a)-\gamma^i(x,t;\varpi^i_{I^j_i})=0,\; i\in N,j\in M_i.
 	\end{array}
 \end{equation}
 Considering that both $\psi_1$ and $\psi_2$ exhibit continuous differentiability, the smooth path outlined in Theorem~\ref{ch3:thm:lgnesmp} undergoes a transformation into a new smooth path derived from System~(\ref{ch3:eqt:lgnetrans}), using a different variable combination.
 
 Let $\widetilde{\mathscr{P}}_{l_n}=\{(x,\nu,t)|(x,\nu,t)\text{ satisfies System~(\ref{ch3:eqt:lgnetrans}) with } 0<t\leq 1\}$ and $\mathscr{P}_{l_n}$ be the closure of $\widetilde{\mathscr{P}}_{l_n}$. The preceding discussion reveals that $\mathscr{P}_{l_n}$ contains a smooth path that originates from the point $(x^*(1),\nu^*(1),1)$ with $x^{*i}(1;\varpi^i_{I^j_i}a)=\gamma^{0i}(\varpi^i_{I^j_i}a)^{1/\kappa_0}-1$ for $i\in N,(j,a)\in D_i$, $x^{*i}(1;\varpi^i_{I^j_i}a)=\gamma^{0i}(\varpi^i_{I^j_i}a)$ for $i\in N,(j,a)\notin D_i$, and $\nu^{*i}_{I^j_i}(1)=0$ for $i\in N,j\in M_i$. As the parameter $t$ approaches $0$, this path converges to a sequence-form Nash equilibrium. 
 
 \subsection{An Alternative Logarithmic-Barrier Smooth Path} 
 
 Inspired by the relationship between realization plans and behavioral strategies, we develop a contrastive logarithmic-barrier artificial game $\Gamma_s^{\tilde l_n}(t)$ in sequence form. In the artificial game $\Gamma_s^{\tilde l_n}(t)$, each player $i$ determines an optimal response to a prescribed strategy $\hat\gamma\in\Lambda$ by solving the strictly convex optimization problem,
 \begin{equation}
 	\label{ch3:opt:lbne}
 	\begin{array}{rl}
 		\max\limits_{\gamma^i} & (1-t)\sum\limits_{j\in M_i}\sum\limits_{a\in A(I^j_i)}\gamma^i(\varpi^i_{I^j_i}a)g^i(\varpi^i_{I^j_i}a,\hat\gamma^{-i})\\
 		&+t\sum\limits_{j\in M_i}\sum\limits_{a\in A(I^j_i)}\gamma^{0i}(\varpi^i_{I^j_i}a)(\ln\gamma^i(\varpi^i_{I^j_i}a)-\ln\gamma^i(\varpi^i_{I^j_i}))\\
 		\text{s.t.} & \sum\limits_{a\in A(I^j_i)}\gamma^i(\varpi^i_{I^j_i}a)-\gamma^i(\varpi^i_{I^j_i})=0,\;j\in M_i.
 	\end{array}
 \end{equation}
 Through the application of the KKT optimality conditions to Problem~(\ref{ch3:opt:lbne}) and the assumption $\hat{\gamma} = \gamma$, we derive the polynomial equilibrium system of $\Gamma_s^{\tilde l_n}(t)$,
 \begin{equation}
 	\label{ch3:eqt:lbne}
 	\begin{array}{l}
 		(1-t)g^i(\varpi^i_{I^j_i}a,\gamma^{-i})+ (1-|M_i(\varpi^i_{I^j_i}a)|)\lambda^i(\varpi^i_{I^j_i}a)\\
 		\hspace{2.8cm}-\nu^i_{I^j_i} + \zeta^i_{I^j_i}(a) = 0,\;i\in N,j\in M_i,a\in A(I^j_i),\\
 		\sum\limits_{a\in A(I^j_i)}\gamma^i(\varpi^i_{I^j_i}a)-\gamma^i(\varpi^i_{I^j_i})=0,\;i\in N,j\in M_i,\\
 		\gamma^i(\varpi^i_{I^j_i}a)\lambda^i(\varpi^i_{I^j_i}a)=t\gamma^{0i}(\varpi^i_{I^j_i}a),\\
 		\hspace{2.4cm}0<\gamma^i(\varpi^i_{I^j_i}a),\;0<\lambda^i(\varpi^i_{I^j_i}a),\;i\in N,j\in M_i,a\in A(I^j_i).
 	\end{array}
 \end{equation}
 By applying the variable substitutions from the first two terms in (\ref{ch3:eqt:lgnesubs}) for all $i\in N,j\in M_i,a\in A(I^j_i)$, and subtracting the $t(1-t)\alpha$ term, we obtain 
 \begin{equation}
 	\label{ch3:eqt:lbnetrans}
 	\begin{array}{l}
 		(1-t)g^i(\varpi^i_{I^j_i}a,\gamma^{-i}(x,t))+ (1-|M_i(\varpi^i_{I^j_i}a)|)\lambda^i(x,t;\varpi^i_{I^j_i}a)\\
 		\hspace{1cm}-\nu^i_{I^j_i} + \zeta^i_{I^j_i}(a) -t(1-t)\alpha(\varpi^i_{I^j_i}a) = 0,\;i\in N,j\in M_i,a\in A(I^j_i),\\
 		\sum\limits_{a\in A(I^j_i)}\gamma^i(x,t;\varpi^i_{I^j_i}a)-\gamma^i(x,t;\varpi^i_{I^j_i})=0,\;i\in N,j\in M_i.
 	\end{array}
 \end{equation}
System (\ref{ch3:eqt:lbnetrans}) defines a smooth path that originates from the point $(x^*(1),\nu^*(1),1)$ with $x^{*i}(1;\varpi^i_{I^j_i}a)=\gamma^{0i}(\varpi^i_{I^j_i}a)^{1/\kappa_0}-1$ for $i\in N,j\in M_i,a\in A(I^j_i)$, and $\nu^{*i}_{I^j_i}(1)=1$ for $i\in N,j\in M_i$. As the parameter $t$ approaches $0$, this path converges to a Nash equilibrium. 

Let $\widetilde{\mathscr{P}}_{\tilde l_n}=\{(x,\nu,t)|(x,\nu,t)\text{ satisfies System~(\ref{ch3:eqt:lbnetrans}) with } 0<t\leq 1\}$ and $\mathscr{P}_{\tilde l_n}$ be the closure of $\widetilde{\mathscr{P}}_{\tilde l_n}$. By applying the transversality theorem and implicit function theorem, we draw the conclusion that there exists a smooth path in $\mathscr{P}_{\tilde l_n}$ for almost any $\alpha\in\mathbb{R}^{n_0}$ with sufficiently small $\|\alpha\|$, which starts from $(x^*(1), \nu^*(1),1)$ at $t = 1$ and converges to a sequence-form Nash equilibrium as $t$ approaches $0$. 

 \section{Numerical Experiments}
 \label{section5}
 Each of System~(\ref{ch3:eqt:lgnetrans}) and~(\ref{ch3:eqt:lbnetrans}) defines a smooth path. By applying the predictor--corrector method, these paths can be numerically traced, leading to the computation of a Nash equilibrium.
 \subsection{Effectiveness Validation}
 \begin{example}{\em
 	In this example, we execute the tracing procedure to solve the games depicted in Fig.~\ref{fig:game1}. The same starting points are employed, and the paths produced by Systems~(\ref{ch3:eqt:lgnetrans}) and (\ref{ch3:eqt:lbnetrans}) are shown in Figs.~\ref{ch3:fig:lgnepathg1}--\ref{ch3:fig:lbnepathg1}, visually illustrating the convergence of the smooth path toward a sequence-form Nash equilibrium.}
 \end{example}
  \begin{figure}[htbp]
 	\centering
 	\begin{minipage}{0.49\textwidth}
 		\centering
 		\includegraphics[width=1\textwidth, height=0.20\textheight]{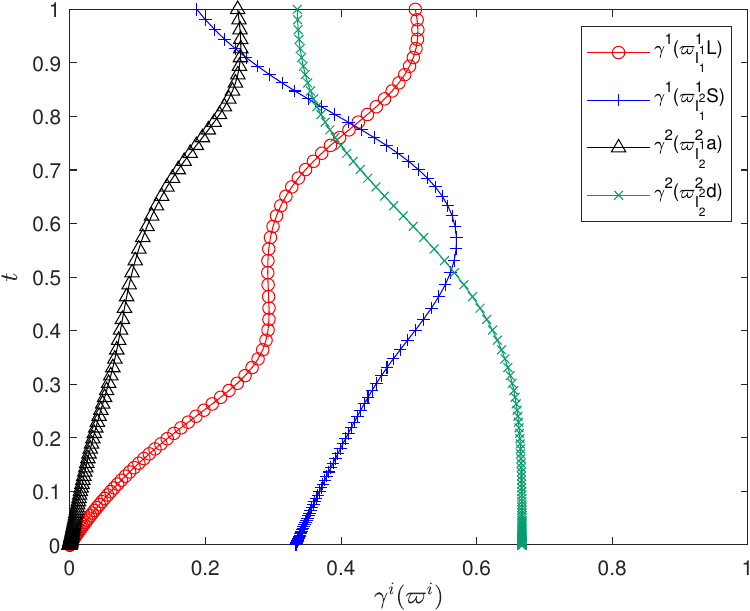}
 		\caption{\label{ch3:fig:lgnepathg1}{\footnotesize Realization Plans on the Path Specified by System~(\ref{ch3:eqt:lgnetrans}) for the Game in Fig.~\ref{fig:game1}}}\end{minipage}\hfill
 	\begin{minipage}{0.49\textwidth}
 		\centering
 		\includegraphics[width=1\textwidth, height=0.20\textheight]{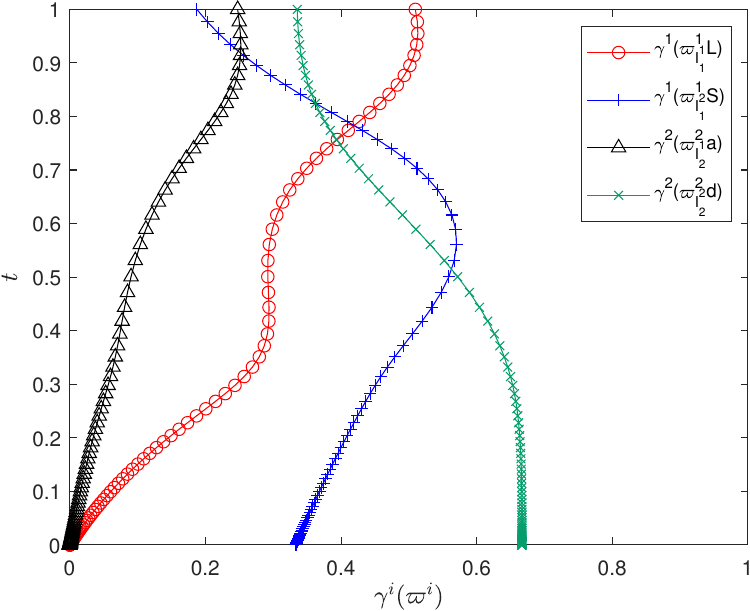}
 		\caption{\label{ch3:fig:lbnepathg1}{\footnotesize Realization Plans on the Path Specified by System~(\ref{ch3:eqt:lbnetrans}) for the Game in Fig.~\ref{fig:game1}}} \end{minipage}
 \end{figure}
  \begin{example}\label{ch2:exm:extnote1} {\em
  		\begin{figure}[H]
  			\centering
  			\begin{minipage}[t]{0.53\textwidth}
  				\centering
  				\includegraphics[width=0.9\textwidth]{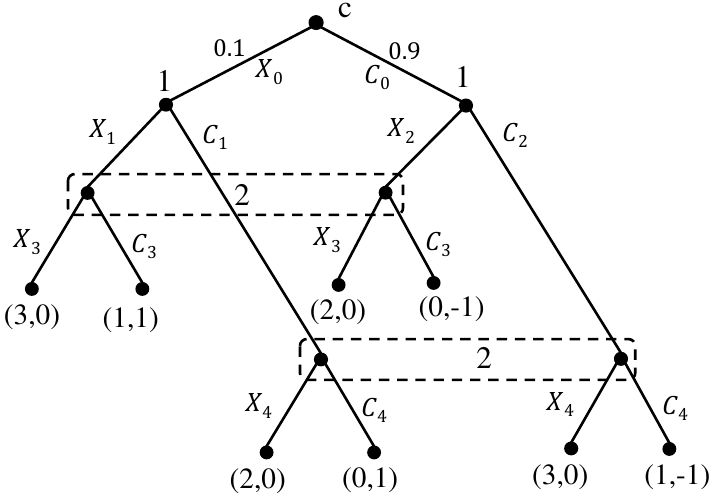}
  				\caption{\label{fig:game2}{\small An Extensive-Form Game from Myerson~\cite{myersonGameTheoryAnalysis1991}}}
  			\end{minipage}\hfill
  			\begin{minipage}[t]{0.44\textwidth}
  				\centering
  				\includegraphics[width=0.9\textwidth]{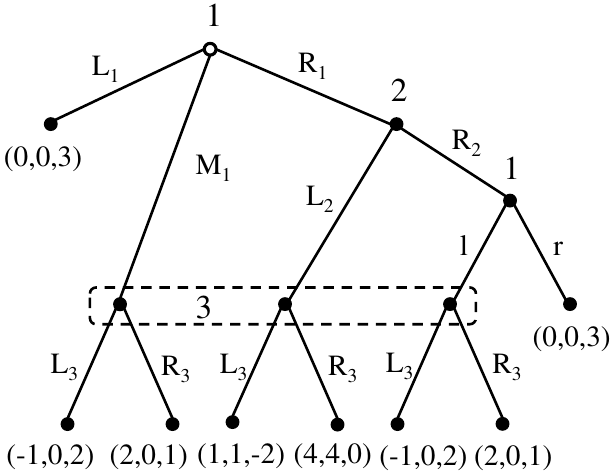}
  				\caption{\label{fig:game3}{\small An Extensive-Form Game from MasColell et al.~\cite{Mas-ColellMicroeconomictheory1995}}}
  			\end{minipage}
  		\end{figure}
  		Consider the extensive-form game $\Gamma$ illustrated in Fig.~\ref{fig:game2}, which replicates the structure of the game presented in Fig.~5.6 of Myerson~\cite{myersonGameTheoryAnalysis1991}. The information sets includes $I^1_c=\{\emptyset\}$, $I^1_1=\{\langle X_0\rangle\}$, $I^2_1=\{\langle C_0\rangle\}$, $I^1_2=\{\langle X_0,X_1\rangle,\langle C_0,X_2\rangle\}$, and $I^1_2=\{\langle X_0,C_1\rangle,\langle C_0,C_2\rangle\}$. The pure strategies available to the chance player are $s^c_1=\{X_0\}$ and $s^c_2=\{C_0\}$, with the corresponding mixed strategy fixed as $\bm{\sigma^c}=(0.1,0.9)^\top$. In the normal-form representation, the effect of chance can be incorporated directly into the payoff computation, thereby simplifying the analysis of pure-strategy profiles of the players. The reduced normal-form representation of this game is shown in Table~\ref{ch2:tab:exm2nf}.
  		\begin{table}[!ht]
  			\centering
  			\caption{Reduced Normal-Form Representation of Fig.~\ref{fig:game2} \label{ch2:tab:exm2nf}}
  			\begin{tabular}{lllll}
  				\toprule
  				& \multicolumn{4}{c}{Player 2 pure strategies} \\
  				\cmidrule(lr){2-5}
  				Player 1 pure strategies 
  				& \(s_1^2=\{X_3,X_4\}\) 
  				& \(s_2^2=\{X_3,C_4\}\) 
  				& \(s_3^2=\{C_3,X_4\}\) 
  				& \(s_4^2=\{C_3,C_4\}\) \\
  				\midrule
  				\(s_1^1=\{X_1,X_2\}\)   & \((2.1,0)\) & \((2.1,0)\) & \((0.1,-0.8)\)  & \((0.1,-0.8)\)  \\
  				\(s_2^1=\{X_1,C_2\}\) & \((3,0)\)  & \((1.2,-0.9)\) & \((2.8,0.1)\)  & \((1,-0.8)\) \\
  				\(s_3^1=\{C_1,X_2\}\) & \((2,0)\)  & \((1.8,0.1)\)  & \((0.2,-0.9)\)  & \((0,-0.8)\)  \\
  				\(s_4^1=\{C_1,C_2\}\) & \((2.9,0)\)  & \((0.9,-0.8)\)  & \((2.9,0)\)  & \((0.9,-0.8)\)  \\
  				\bottomrule
  			\end{tabular}
  		\end{table}
  		the paths produced by Systems~(\ref{ch3:eqt:lgnetrans}) and (\ref{ch3:eqt:lbnetrans}) are shown in Figs.~\ref{ch3:fig:lgnepathg2}--\ref{ch3:fig:lbnepathg2}, visually illustrating the convergence of the smooth path toward a mixed-strategies Nash equilibrium.
  	} 
  \end{example}
   \begin{figure}[htbp]
  	\centering
  	\begin{minipage}{0.49\textwidth}
  		\centering
  		\includegraphics[width=1\textwidth, height=0.20\textheight]{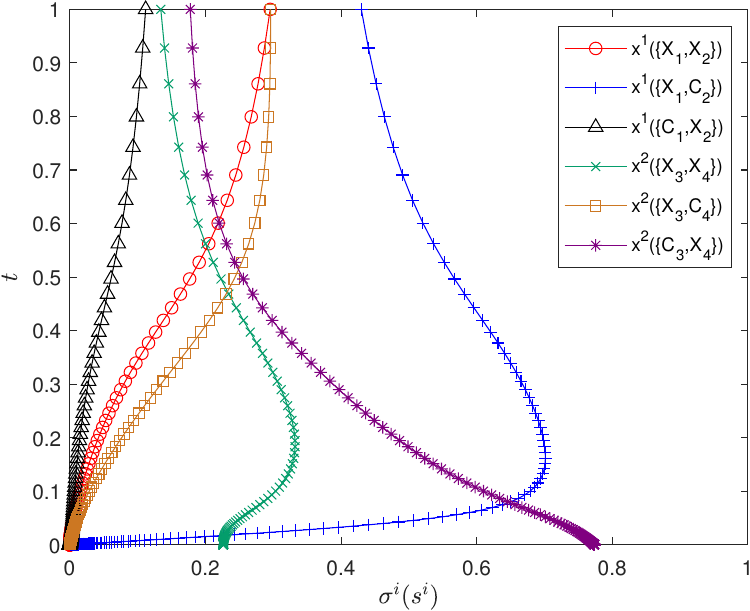}
  		\caption{\label{ch3:fig:lgnepathg2}{\footnotesize Mixed Strategies on the Path Specified by System~(\ref{ch3:eqt:lgnetrans}) for the Game in Fig.~\ref{fig:game2}}}\end{minipage}\hfill
  	\begin{minipage}{0.49\textwidth}
  		\centering
  		\includegraphics[width=1\textwidth, height=0.20\textheight]{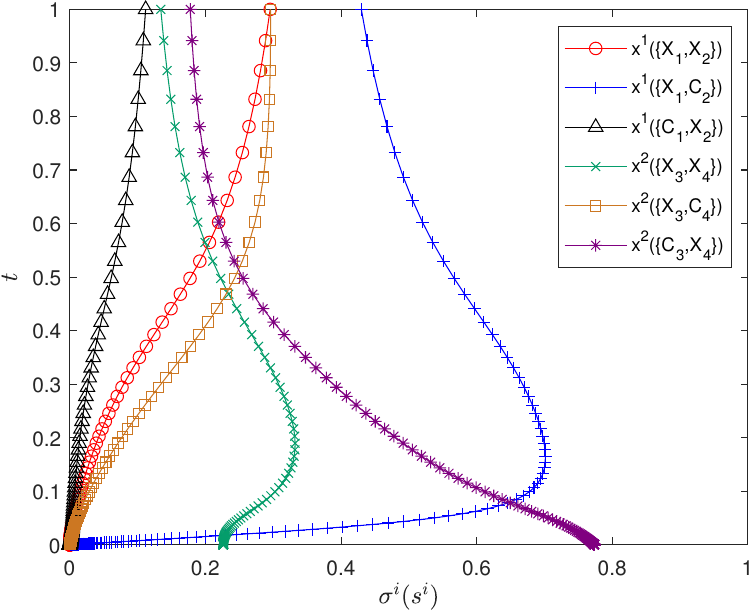}
  		\caption{\label{ch3:fig:lbnepathg2}{\footnotesize Mixed Strategies on the Path Specified by System~(\ref{ch3:eqt:lbnetrans}) for the Game in Fig.~\ref{fig:game2}}} \end{minipage}
  \end{figure}
  \begin{example}\label{ch2:exm:extnote2} {\em
  		Consider another extensive-form game $\Gamma$ shown in Fig.~\ref{fig:game3}, which is the game in Fig.~9.C.2 of MasColell et al.~\cite{Mas-ColellMicroeconomictheory1995}. The players' information sets are given by \( I^1_1=\{\emptyset\} \), \( I^2_1=\{\langle R_1,R_2\rangle\} \), \( I^1_2=\{\langle R_1\rangle\} \), and \( I^1_3=\{\langle M_1\rangle,\langle R_1,L_2\rangle,\langle R_1,R_2,l\rangle\} \). The reduced normal-form representation of the extensive-form game can be summarized in Tab.~\ref{ch2:tab:exm3nf}.
  		\begin{table}[!ht]
  			\centering
  			\caption{Reduced Normal-Form Representation of Fig.~\ref{fig:game3} \label{ch2:tab:exm3nf}}
  			\begin{tabular}{lllll}
  				\toprule
  				\multirow{2}{*}{\diagbox{$S^1$}{$S^2$ and $S^3$}}
  				& \multicolumn{2}{c}{$s^2_1=\{L_2\}$}
  				& \multicolumn{2}{c}{$s^2_2=\{R_2\}$} \\
  				\cmidrule(lr){2-3}\cmidrule(lr){4-5}
  				& $s^3_1=\{L_3\}$ & $s^3_2=\{R_3\}$
  				& $s^3_1=\{L_3\}$ & $s^3_2=\{R_3\}$ \\
  				\midrule
  				$s^1_1=\{L_1\}$ & $(0,0,3)$ & $(0,0,3)$ & $(0,0,3)$ & $(0,0,3)$ \\
  				$s^1_2=\{M_1\}$ & $(-1,0,2)$ & $(2,0,1)$ & $(-1,0,2)$ & $(2,0,1)$ \\
  				$s^1_3=\{R_1,l\}$ & $(1,1,-2)$ & $(4,4,0)$ & $(-1,0,2)$ & $(2,0,1)$ \\
  				$s^1_3=\{R_1,r\}$ & $(1,1,-2)$ & $(4,4,0)$ & $(0,0,3)$ & $(0,0,3)$ \\
  				\bottomrule
  			\end{tabular}
  		\end{table}
  		the paths produced by Systems~(\ref{ch3:eqt:lgnetrans}) and (\ref{ch3:eqt:lbnetrans}) are shown in Figs.~\ref{ch3:fig:lgnepathg3}--\ref{ch3:fig:lbnepathg3}, visually illustrating the convergence of the smooth path toward a mixed-strategies Nash equilibrium.
  	}
  \end{example}
 \begin{figure}[htbp]
 	\centering
 	\begin{minipage}{0.49\textwidth}
 		\centering
 		\includegraphics[width=1\textwidth, height=0.20\textheight]{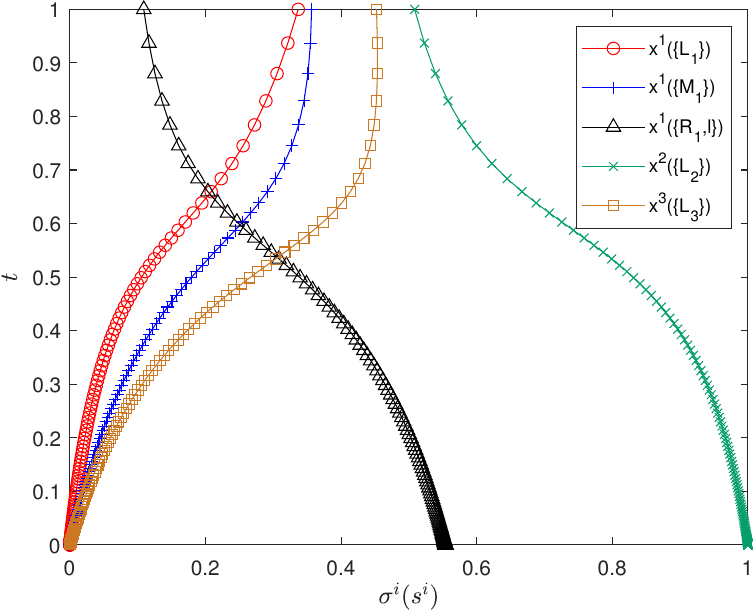}
 		\caption{\label{ch3:fig:lgnepathg3}{\footnotesize Mixed Strategies on the Path Specified by System~(\ref{ch3:eqt:lgnetrans}) for the Game in Fig.~\ref{fig:game3}}}\end{minipage}\hfill
 	\begin{minipage}{0.49\textwidth}
 		\centering
 		\includegraphics[width=1\textwidth, height=0.20\textheight]{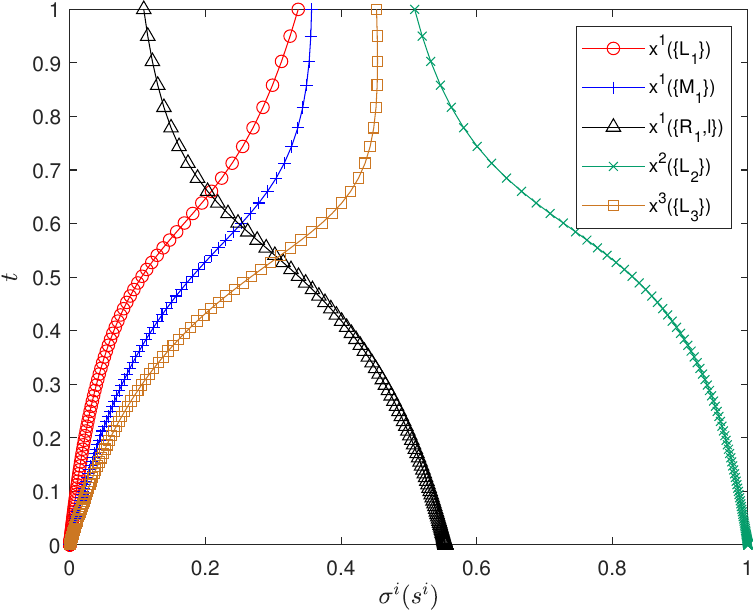}
 		\caption{\label{ch3:fig:lbnepathg3}{\footnotesize Mixed Strategies on the Path Specified by System~(\ref{ch3:eqt:lbnetrans}) for the Game in Fig.~\ref{fig:game3}}} \end{minipage}
 \end{figure}
 \subsection{Performance Evaluation}
 To compare the convergence performance of our methods, we employ two structurally distinct types of random extensive-form games, as shown in Figures~\ref{ch3:fig:rg1}--\ref{ch3:fig:rg2}. Both game types are parameterized by the number of players ($n$), the maximum historical depth ($\mathcal{L}$), and the number of allowable actions per information set ($\mathcal{A}$). In these games, players act cyclically, with the terminal payoffs determined by random integers uniformly distributed between $-10$ and $10$. A detailed explanation of the two game types is provided below.
 		\begin{itemize}
 			\item \textbf{Type 1:} As shown in Figures~\ref{ch3:fig:rg1}, histories are classified into the same information set only when they diverge in the final actions taken. Moreover, all terminal histories exhibit an identical length.
 			\item \textbf{Type 2:} As represented in Figures~\ref{ch3:fig:rg2}, this structural configuration is commonly found in the literature. For odd-indexed players, each information set consists of a single history. In contrast, for even-indexed players, histories are grouped into the same information set only when they share an identical corresponding sequence. The probability that player $0$ chooses each of the available actions is equal, and the total number of actions is fixed at $3$, without loss of generality.
 		\end{itemize}
 		Since the number of players does not directly impact the game size, we set $n=3$ for Type 1 games and $n=4$ for Type 2 games, adjusting the other two parameters to control the game size. To realize a comprehensive comparative analysis of the three path-following methods, $20$ random games with distinct payoffs were generated and solved for each parameter configuration $(\mathcal{L}, \mathcal{A})$ in both game types. Each instance was initialized from a randomly generated starting point, and the parameters of the predictor–corrector algorithm were kept consistent throughout all experiments. The predictor step size was set to $0.05t^{0.3}$, and the correction tolerance was set to $0.5t^{0.3}$. A run was considered successful if the termination criterion $t<10^{-4}$ was satisfied; otherwise, it was deemed a failure if either the iteration count or the computation time exceeded a prescribed limit. All computations were conducted on a Windows Server 2016 Standard system equipped with two Intel(R) Xeon(R) CPU E5-2650 v4 @ 2.20GHz processors and 128GB of RAM. The numerical results in Tables~\ref{t2} and \ref{t3}.
 		\begin{figure}[H]
 			\centering
 			\begin{minipage}[b]{0.55\textwidth}
 				\centering
 				\includegraphics[width=0.95\textwidth]{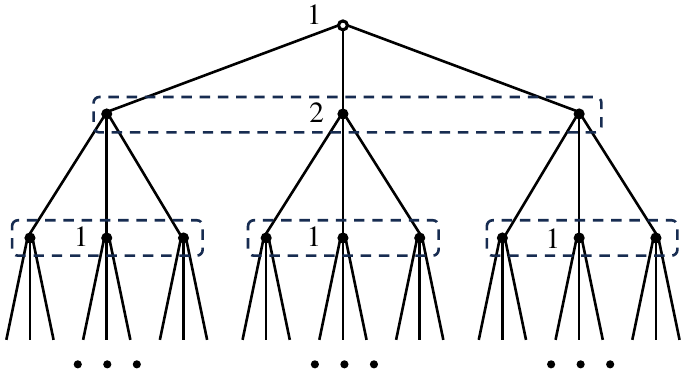}
 				\caption{\label{ch3:fig:rg1}{\small A Random Extensive-Form Game of Type 1}}
 			\end{minipage}\hfill
 			\begin{minipage}[b]{0.43\textwidth}
 				\centering
 				\includegraphics[width=0.95\textwidth]{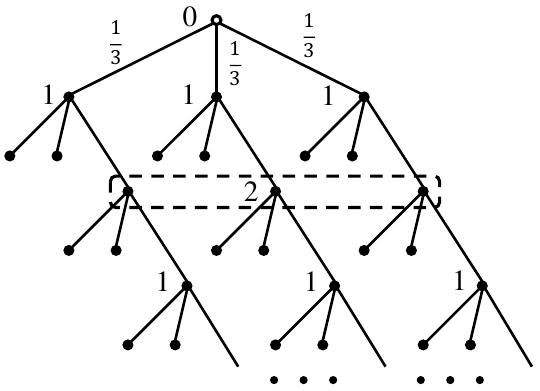}
 				\caption{\label{ch3:fig:rg2}{\small A Random Extensive-Form Game of Type 2}}
 			\end{minipage}
 	\end{figure}
 \begin{table}[htbp]\centering\renewcommand\arraystretch{0.95}
 	\caption{Numerical Comparisons for the Game in Fig.~\ref{ch3:fig:rg1}}\label{t2}
 	\begin{tabular*}{\textwidth}{@{\extracolsep\fill}>{\rowmac}l>{\rowmac}l>{\rowmac}l>{\rowmac}l>{\rowmac}l>{\rowmac}l>{\rowmac}l>{\rowmac}l>{\rowmac}l<{\clearrow}}\toprule
 		&  &  & \multicolumn{2}{l}{Iteration Numbers} & \multicolumn{2}{l}{Computational Time} & \multicolumn{2}{l}{Failure Rates} \\\cmidrule(r){4-5}\cmidrule(r){6-7}\cmidrule(r){8-9}
 		$(n,\mathcal{L},\mathcal{A})$ & Dim & & LGNE & LBNE & LGNE & LBNE & LGNE & LBNE\\\midrule
 		$(3,5,2)$ & $49$ & max & 316 & 242 & 9.9 & 10.0 & 0\% & 0\%\\
 		& & min & 121 & 136 & 4.0 & 4.4 & & \\
 		& &\setrow{\bfseries} med & 156.5 & 156.5 & 5.2 & 5.2 & &\\
 		$(3,6,2)$ & $97$ & max & 217 & 226 & 69.1 & 38.2 & 0\% & 0\%\\
 		& & min & 113 & 148 & 17.1 & 19.0 & & \\
 		& &\setrow{\bfseries} med & 166.0 & 178.0 & 21.4 & 23.4 & &\\
 		$(3,7,2)$ & $193$ & max & 592 & 442 & 372.1 & 352.8 & 0\% & 0\%\\
 		& & min & 171 & 74 & 98.4 & 105.1 & & \\
 		& &\setrow{\bfseries} med & 239.5 & 197.0 & 157.2 & 121.8 & &\\
 		$(3,8,2)$ & $385$ & max & - & - & - & - & 25\% & 15\%\\
 		& & min & 192 & 58 & 720.2 & 693.5 & & \\
 		& &\setrow{\bfseries} med & 278.5 & 226.5 & 1048.0 & 988.9 & &\\
 		$(3,4,3)$ & $57$ & max & 348 & 264 & 34.7 & 23.0 & 0\% & 0\%\\
 		& & min & 138 & 134 & 11.3 & 11.6 & & \\
 		& &\setrow{\bfseries} med & 188.0 & 186.5 & 16.0 & 16.1 & &\\
 		$(3,4,4)$ & $111$ & max & 491 & 341 & 515.2 & 374.3 & 0\% & 0\%\\
 		& & min & 161 & 158 & 91.6 & 95.2 & & \\
 		& &\setrow{\bfseries} med & 242.5 & 201.5 & 144.7 & 123.6 & &\\
 		$(3,4,5)$ & $193$ & max & - & - & - & - & 10\% & 20\%\\
 		& & min & 181 & 164 & 539.8 & 494.1 & & \\
 		& &\setrow{\bfseries} med & 233.5 & 239.0 & 708.6 & 748.5 & &\\
 		\bottomrule
 	\end{tabular*}
 \end{table}
 \begin{table}[htbp]\centering\renewcommand\arraystretch{0.95}
 	\caption{Numerical Comparisons for the Game in Fig.~\ref{ch3:fig:rg2}}\label{t3}
 	\begin{tabular*}{\textwidth}{@{\extracolsep\fill}>{\rowmac}l>{\rowmac}l>{\rowmac}l>{\rowmac}l>{\rowmac}l>{\rowmac}l>{\rowmac}l>{\rowmac}l>{\rowmac}l<{\clearrow}}\toprule
 		&  &  & \multicolumn{2}{l}{Iteration Numbers} & \multicolumn{2}{l}{Computational Time} & \multicolumn{2}{l}{Failure Rates} \\\cmidrule(r){4-5}\cmidrule(r){6-7}\cmidrule(r){8-9}
 		$(n,\mathcal{L},\mathcal{A})$ & Dim & & LGNE & LBNE & LGNE & LBNE & LGNE & LBNE\\\midrule
 		$(4,10,2)$ & $61$ & max & 177 & 175 & 9.3 & 9.0 & 0\% & 0\%\\
 		& & min & 82 & 84 & 4.5 & 4.5 & & \\
 		& &\setrow{\bfseries} med & 105.5 & 107.0 & 5.8 & 5.9 & &\\
 		$(4,20,2)$ & $121$ & max & 163 & 178 & 41.7 & 45.6 & 0\% & 0\%\\
 		& & min & 84 & 97 & 20.2 & 23.1 & & \\
 		& &\setrow{\bfseries} med & 102.5 & 110.5 & 24.3 & 26.5 & &\\
 		$(4,30,2)$ & $181$ & max & 263 & 302 & 147.3 & 169.7 & 0\% & 0\%\\
 		& & min & 93 & 101 & 55.5 & 60.9 & & \\
 		& &\setrow{\bfseries} med & 115.0 & 119.5 & 70.2 & 72.7 & &\\
 		$(4,40,2)$ & $241$ & max & 147 & 158 & 193.5 & 207.8 & 0\% & 0\%\\
 		& & min & 91 & 106 & 121.6 & 140.7 & & \\
 		& &\setrow{\bfseries} med & 113.0 & 119.0 & 150.5 & 159.2 & &\\
 		$(4,10,4)$ & $101$ & max & 167 & 166 & 42.8 & 43.1 & 0\% & 0\%\\
 		& & min & 94 & 103 & 24.1 & 26.6 & & \\
 		& &\setrow{\bfseries} med & 110.5 & 115.5 & 28.4 & 30.0 & &\\
 		$(4,10,6)$ & $141$ & max & 244 & 262 & 153.4 & 165.1 & 0\% & 0\%\\
 		& & min & 114 & 120 & 73.9 & 77.8 & & \\
 		& &\setrow{\bfseries} med & 124.0 & 128.5 & 80.7 & 84.5 & &\\
 		$(4,10,8)$ & $181$ & max & 142 & 143 & 179.7 & 183.3 & 0\% & 0\%\\
 		& & min & 121 & 127 & 154.5 & 166.1 & & \\
 		& &\setrow{\bfseries} med & 130.0 & 133.5 & 165.9 & 171.5 & &\\
 		$(4,10,10)$ & $221$ & max & 144 & 147 & 327.4 & 334.1 & 0\% & 0\%\\
 		& & min & 131 & 135 & 290.7 & 306.6 & & \\
 		& &\setrow{\bfseries} med & 136.5 & 142.0 & 307.6 & 322.0 & &\\
 		\bottomrule
 	\end{tabular*}
 \end{table}
 \section{Conclusion}
 \label{section6}
 This paper investigates the computation of Nash equilibria in finite $n$-player extensive-form games with perfect recall under the sequence-form representation. We depart from the conventional view of the sequence form as merely a computational reformulation and instead develop a direct sequence-form characterization of Nash equilibrium. By establishing a rigorous equivalence with mixed-strategy Nash equilibrium, we derive a well-defined sequence-form Nash equilibrium system, providing a stronger theoretical foundation for equilibrium analysis. On this basis, we develop a single-stage interior-point differentiable path-following method for solving the equilibrium system. The proposed method employs logarithmic-barrier regularization to generate a smooth solution path within the interior of the realization-plan space, yielding favorable numerical stability and convergence properties. Computational results confirm the effectiveness and efficiency of the approach. The proposed framework bridges equilibrium characterization and computation in the sequence-form setting. Future work may consider extensions to equilibrium refinements, improved scalability for large-scale instances, and broader applications in multi-agent decision-making problems.
 
 \begin{appendices}
\section{The Boundedness of $\widetilde{\mathscr{S}}_L$ in Theorem~\ref{ch3:thm:lgnecnt}}
\label{app:compactness_sl}

The objective of this appendix is to elucidate the boundedness of $\widetilde{\mathscr{S}}_L$, which is necessary for proving Theorem~\ref{ch3:thm:lgnecnt}.

Let $(\gamma^*, \lambda^*, \nu^*,t)\in \mathscr{S}_L$ be a solution to System~(\ref{ch3:eqt:lgne}). By applying the backward induction method to the first group of equations in System~(\ref{ch3:eqt:lgne}), we derive the following equations for $i\in N$, $j\in M_i$, and $a\in A(I^j_i)$,
\begin{equation}\label{app1:recs}
	\begin{array}{l}
		-\nu^{*i}_{I^j_i}+\sum\limits_{j_q\in M_i,a_q\in A(I^{j_q}_i),a\in\varpi^i_{I^{j_q}_i}a_q}\frac{\gamma^{*i}(\varpi^i_{I^{j_q}_i}a_q)}{\gamma^{*i}(\varpi^i_{I^{j}_i}a)}(1-t)g^i(\varpi^i_{I^{j_q}_i}a_q,\gamma^{*-i})\\
		\hspace{3.5cm}+\sum\limits_{(j_q,a_q)\in D_i,a\in\varpi^i_{I^{j_q}_i}a_q}\frac{\gamma^{*i}(\varpi^i_{I^{j_q}_i}a_q)}{\gamma^{*i}(\varpi^i_{I^{j}_i}a)}(\lambda^{*i}(\varpi^i_{I^{j_q}_i}a_q)-t)=0.
	\end{array}
\end{equation}
Specifically, for $i\in N$, $(j,a)\in D_i$, the equations (\ref{app1:recs}) can be directly obtained from the first group of equations in System (\ref{ch3:eqt:lgne}). In cases where $(j,a)\notin D_i$, we assume that Equations (\ref{app1:recs}) hold for all $j_0\in M_i(\varpi^i_{I^j_i}a), a_0\in A(I^{j_0}_i)$. By multiplying both sides of Equations (\ref{app1:recs}) by $\gamma^*(\varpi^i_{I^{j_0}_i}a_0)/\gamma^*(\varpi^i_{I^{j_0}_i})$ and summing over $a_0\in A(I^{j_0}_i)$, we derive the expression for $\nu^i_{I^{j_0}_i}$. Finally, by substituting $\zeta^i_{I^j_i}(a)$ with the recursive outcomes, we obtain Equations (\ref{app1:recs}) for $i\in N$, $(j,a)\notin D_i$. Let $U^i_l=\min_{h\in Z}u^i(h)$ and $U^i_u=\max_{h\in Z}u^i(h)$. Equations (\ref{app1:recs}) imply that $\nu^i_{I^j_i}\geq-|W^i||U^i_l|-|D_i|$ for any $i\in N,j\in M_i$. Based on the above derivations, we analyze the upper bound of $(\lambda^*,\nu^*)$.

For $i\in N,j\in M_i$ with $\varpi^i_{I^j_i}=\emptyset$, we multiply $\gamma^{*i}(\varpi^i_{I^{j}_i}a)$ on both sides of Equations (\ref{app1:recs}) and sum over $a\in A(I^j_i)$, obtaining
\begin{equation}\label{app1:recsroot}
	\begin{array}{l}
		-\nu^{*i}_{I^j_i}+\sum\limits_{j_q\in M_i,a_q\in A(I^{j_q}_i),\varpi^i_{I^{j_q}_i}a_q\cap A(I^i_j)\neq\emptyset}\gamma^{*i}(\varpi^i_{I^{j_q}_i}a_q)(1-t)g^i(\varpi^i_{I^{j_q}_i}a_q,\gamma^{*-i})\\
		\hspace{3.0cm}+\sum\limits_{(j_q,a_q)\in D_i,\varpi^i_{I^{j_q}_i}a_q\cap A(I^i_j)\neq\emptyset}\gamma^{*i}(\varpi^i_{I^{j_q}_i}a_q)(\lambda^{*i}(\varpi^i_{I^{j_q}_i}a_q)-t)=0.
	\end{array}
\end{equation}
Given that $\gamma^{*i}(\varpi^i_{I^{j_q}_i}a_q)\lambda^{*i}(\varpi^i_{I^{j_q}_i}a_q)\leq\gamma^{0i}(\varpi^i_{I^{j}_i}a)< 1$, it follows from Equations~(\ref{app1:recsroot}) that $\nu^{*i}_{I^j_i}\leq|W^i||U^i_u|\triangleq V^j_i$. From the second group of equations in System (\ref{ch3:eqt:lgne}), we further obtain $\zeta^i_{I^j_i}(a)\leq V^j_i$. Consequently, for any $a\in A(I^j_i)$ and $j_0\in M_i(\varpi^i_{I^j_i}a)$, we have $\nu^{*i}_{I^{j_0}_i}\leq V^j_i+(|M_i(\varpi^i_{I^j_i}a)|-1)(|W^i||U^i_l|+|D_i|)\triangleq V^{j_0}_i$. This reasoning extends recursively through forward induction. Since the game is finite, $\nu^{*i}_{I^j_i}$ is bounded above by $V_u=\max_{i_q\in N,j_q\in M_i}V^{j_q}_{i_q}$ for any $i\in N$ and $j\in M_i$. Hence, from the first group of equations in System (\ref{ch3:eqt:lgne}), it follows that $\lambda^{*i}(\varpi^i_{I^j_i}a)\leq V_u+|U^i_l|+1$ for $i\in N$, $(j,a)\in D_i$.

\section{Jacobian Matrix of $p(\cdot)$ Full-Row Rank Proof in Theorem~\ref{ch3:thm:lgnesmp}}
\label{app:fullrowrank}

This appendix proves that the Jacobian matrix $Dp(\gamma,\lambda,\nu,t;\alpha)$ of $p(\gamma,\lambda,\nu,t;\alpha)$ has full row rank on $\text{int}(\Lambda)\times\mathbb{R}^{n_0}_{++}\times\mathbb{R}^{m_0}\times(0,1)\times\mathbb{R}^{n_0}$, which is critical for the proof of Theorem~\ref{ch3:thm:lgnesmp}. \par
Consider the case where $t\in(0,1)$. We denote the first $n_0$ terms of $p(\gamma,\lambda,\nu,t;\alpha)$ as $g(\gamma,\lambda,\nu,t;\alpha)$. The Jacobian matrix $Dp(\gamma,\lambda,\nu,t;\alpha)$ is given by
\begin{equation}
	Dp(\gamma,\lambda,\nu,t;\alpha)=\left(\begin{array}{ccccc}
		D_{\gamma} g	& D_\lambda g	& D_\nu g	& D_t g		& -t(1-t)I^{n_0\times n_0} \\
		B				& 0				& 0			& 0			& 0\\
		C				& E				& 0			& F			& 0
	\end{array}\right),
	\nonumber
\end{equation}
where $I^{n_0\times n_0}$ is an $n_0\times n_0$ identity matrix, $B=\bar B-\tilde B$, 
\begin{equation}
	\bar B=\left(\begin{array}{cccc}
		{e^1_1}^\top&&&\\
		&{e^2_1}^\top&&\\
		&&\ddots&\\
		&&&{e^{m_n}_n}^\top
	\end{array}\right)\in \mathbb{R}^{m_0\times n_0} \text{ with } e^{j}_i=(1,1,\ldots,1)^\top\in\mathbb{R}^{|A(I^j_i)|}.
	\nonumber
\end{equation}
The matrix $\tilde B\in \mathbb{R}^{m_0\times n_0}$ is defined such that, in each row, the element corresponding to the sequence associated with the relevant information set takes the value $1$, whereas all remaining elements are set to $0$. The matrix $C\in \mathbb{R}^{|D_i|\times n_0}$ assigns $\lambda^i(\varpi^i_{I^j_i}a)$ to the element whose row and column both correspond to the sequence $\varpi^i_{I^j_i}a$, $(j,a)\in D_i$. $E\in \mathbb{R}^{|D_i|\times |D_i|}$ is a diagonal matrix with its element corresponding to the sequence $\varpi^i_{I^j_i}a$ defined as $\gamma^i(\varpi^i_{I^j_i}a)$. The vector $F\in\mathbb{R}^{|D_i|}$ is the partial derivative of the left-hand side of the third equation group in System~(\ref{ch3:eqt:lgne}) with respect to $t$. We observe that $I^{n_0\times n_0}$, $B$ and $E$ are of full-row rank. Thus, for any $t\in(0,1)$, the Jacobian matrix $Dp(\gamma,\lambda,\nu,t;\alpha)$ is of full-row rank.\par
When $t=1$, one can see that System~(\ref{ch3:eqt:lgne}) is reduced to System~(\ref{ch3:eqt:lgnests}), the Jacobian matrix becomes
\begin{equation}
	Dp(\gamma,\lambda,\nu,1;\alpha)=\left(\begin{array}{ccccc}
		0	& 0			& I^{|D_i|\times |D_i|}		& -B^{\top}_1		& 0\\
		0	& 0			& 0							& -B^{\top}_2		& 0\\
		B_1	& B_2		& 0							& 0				& 0\\
		C_0	& 0			& E							& 0				& 0 
	\end{array}\right). 
	\nonumber
\end{equation}
where the matrix $B$ is column-wise reordered into $(B_1, B_2)$, where $B_1$ includes columns associated with sequences $\varpi^i_{I^j_i}a$ for $(j,a)\in D_i$, and $B_2$ includes those for $(j,a)\notin D_i$. It can be observed that $B_2$ has full column rank. $C_0\in \mathbb{R}^{|D_i|\times |D_i|}$ is a diagonal matrix whose element, corresponding to the sequence $\varpi^i_{I^j_i}a,(j,a)\in D_i$, equal to $\lambda^i(\varpi^i_{I^j_i}a)$. Applying row and column operations, one can transform $Dp(\gamma,\lambda,\nu,1;\alpha)$ to 
\begin{equation}
	Dp(\gamma,\lambda,\nu,1;\alpha)=\left(\begin{array}{ccccc}
		C_0E^{-1}	& 0			& 0					& 0		& 0\\
		0			& 0			& -B^{\top}_2			& 0		& 0\\
		0			& B_2		& B_1EC_0^{-1}B^{\top}_1	& 0		& 0\\
		0			& 0			& 0					& E		& 0 
	\end{array}\right). 
	\nonumber
\end{equation}
Let  
$
M = 
\begin{pmatrix}
	0 & -B_2^{\top} \\[4pt]
	B_2 & B_1 E C_0^{-1} B_1^{\top}
\end{pmatrix},
\;
\begin{pmatrix} x \\ y \end{pmatrix}
\text{ satisfies }
M 
\begin{pmatrix} x \\ y \end{pmatrix} = 0.
$ Expanding the block equations yields  
\[
\begin{cases}
	-B_2^{\top} y = 0, \\[4pt]
	B_2 x + B_1 E C_0^{-1} B_1^{\top} y = 0.
\end{cases}
\]
From the first equation, we obtain \(B_2^{\top} y = 0\). Premultiplying the second equation by \(y^{\top}\) gives  
\[
y^{\top} B_2 x + y^{\top} B_1 E C_0^{-1} B_1^{\top} y = 0.
\]
Since \(B_2^{\top} y = 0\), the first term vanishes, and thus \(y^{\top} B_1 E C_0^{-1} B_1^{\top} y = 0\). Because \(E C_0^{-1}\) is positive definite, it follows that \(B_1^{\top} y = 0\). Consequently,
\[
B_1^{\top} y = 0, \qquad B_2^{\top} y = 0.
\]
Given that \((B_1, B_2)\) has full row rank, we conclude that \(y = 0\). Substituting \(y = 0\) into the second equation gives \(B_2 x = 0\). As \(B_2\) is of full column rank, it follows that \(x = 0\). Therefore, \(M\) is nonsingular, implying that \(Dp(\gamma,\lambda,\nu,1;\alpha)\) has full row rank.
\end{appendices}
\newpage
\bibliography{library}
\end{document}